\documentclass[12pt,draftcls,onecolumn]{IEEEtran}
\usepackage{stfloats}
\usepackage{amsfonts}
\usepackage{amssymb}
\usepackage{amsthm}
\usepackage{cite}
\usepackage[cmex10]{amsmath}
\usepackage{float}
\usepackage{color}
\usepackage{stfloats,fancyhdr}
\usepackage{amsmath,bm}
\usepackage{algorithm}
\usepackage{algorithmic}
\usepackage{multirow}
\usepackage{changepage}
\usepackage[normalem]{ulem}
\usepackage{amsthm}
\usepackage{balance}

\newtheorem{theorem}{Theorem}
\newtheorem{lemma}{Lemma}
\newtheorem{proposition}{Proposition}

\newtheorem{assumption}{Assumption}
\newtheorem{remark}{Remark}

\ifCLASSINFOpdf
\usepackage[pdftex]{graphicx}
\DeclareGraphicsExtensions{.pdf,.jpeg,.png}
\else
\usepackage[dvips]{graphicx}
\DeclareGraphicsExtensions{.eps}
\fi

\usepackage{subfigure}
\usepackage{fancybox,dashbox}
\usepackage{authblk}
\usepackage{tabularx}

\begin{document}
		
\title{Real-Time Remote Estimation with Hybrid~ARQ in Wireless Networked Control}

\author{\IEEEauthorblockN{Kang Huang, Wanchun Liu$^\dagger$, Mahyar Shirvanimoghaddam, Yonghui Li and Branka Vucetic}
	\thanks{\setlength{\baselineskip}{13pt} \noindent The authors are with School of Electrical and Information Engineering, The University of Sydney, Sydney NSW 2006, Australia.
		(emails: \{kang.huang,\ wanchun.liu,\ mahyar.shirvanimoghaddam,\ yonghui.li,\ branka.vucetic\}@sydney.edu.au).
		$^\dagger${Wanchun Liu is the corresponding author}.
		Part of the paper will be presented at Proc. IEEE ICC'2019~\cite{Kang2019ICC}.
	}
}

\maketitle
\vspace{-1cm}
\begin{abstract}
	Real-time remote estimation is critical for mission-critical applications including industrial automation, smart grid and tactile Internet.
	In this paper, we propose a hybrid automatic repeat request (HARQ)-based real-time remote estimation framework for linear time-invariant (LTI) dynamic systems.
	Considering the estimation quality of such a system, there is a fundamental tradeoff between the reliability and freshness of the sensor's measurement transmission.
%	When a transmission fails, the sensor can either retransmit the previous (i.e., old) measurement such that the receiver can obtain a more reliable old measurement, or transmit a new but less reliable measurement.
	We formulate a new problem to optimize the sensor's online transmission control policy for static and Markov fading channels, which depends on both the current estimation quality of the remote estimator and the current number of retransmissions of the sensor, so as to minimize the long-term remote estimation mean squared error (MSE).
	This problem is non-trivial.
	In particular, it is challenging to derive the condition in terms of the communication channel quality and the LTI system parameters, to ensure a bounded long-term estimation MSE.
	We derive a sufficient condition of the existence of a stationary and deterministic optimal policy that stabilizes the remote estimation system and minimizes the MSE.
	Also, we prove that the optimal policy has a switching structure, and accordingly derive a low-complexity suboptimal policy.
	Numerical results show that the proposed optimal policy significantly improves the performance of the remote estimation system compared to the conventional non-HARQ policy.
	
\end{abstract}
\newpage
	\section{Introduction}	
	Real-time remote estimation is critical for networked control applications such as
	industrial automation, smart grid, vehicle platooning, drone swarming, immersive virtual reality (VR) and the tactile Internet~\cite{Antonakoglou2018towards}.
	For such real-time applications, high quality remote estimation of the states of dynamic processes over unreliable links is a major challenge.
	The sensor's sampling policy, the estimation scheme at a remote receiver, and the communication protocol for state-information delivery between the sensor and the receiver should be designed jointly.
	%(e.g., sensor networks, smart grid, airplane/vehicular control, robotics, Internet of Things, and cyber-physical systems). 

	%Mission-critical machine-type communication is starting to play a central role in industrial Internet of things, which imposes requirements such as lower power, high reliability and lower latency \cite{orsino2017effects}. Wireless networked control system, which allows systems to be remotely operated, is a important example of mission-critical applications. Although wireless techniques bring convenience to system deployment, it is a challenge to design efficient communication protocol.
	
	To enable the optimal design of wireless {remote estimation}, the performance metric for the remote estimation system needs to be selected properly. 
	For some applications, the model of the dynamic process under monitoring is unknown and the receiver is not able to estimate the current state of the process based on the previously received states, i.e., a state-monitoring-only scenario~\cite{kaul2012real}. In this scenario, the performance metric is the age-of-information (AoI), which reflects how old the freshest received sensor measurement is, since the moment that measurement was generated at the sensor~\cite{kaul2012real}.
	However, in practice, most of the dynamic processes are time-correlated, and the state-changing rules can be known by the receiver to some extent. 
	Therefore, the receiver can estimate the current state of the process based on the previously received measurements and the model of the dynamic process (see e.g., \cite{schenato2008optimal,sun2017remote}), especially when the transmission of the packet that carries the current sensor measurement is failed or delayed. 
	In this sense, the estimation mean squared error (MSE) is the appropriate performance metric.	
	
%	To enable optimal design of wireless remote estimation, the performance metric for the remote estimation system need to be selected properly. There are in general two types of performance measurements: the age-of-information (AoI)~\cite{kaul2012real}, i.e., an information-freshness related one, and the estimation mean-square error (MSE)~\cite{schenato2008optimal}, i.e., a information-accuracy related one.
%	The AoI reflects how old the freshest received sensor measurement is, since the moment that measurement was generated at the sensor~\cite{kaul2012real}.
%	However, the AoI cannot capture how fast the information changes~\cite{sun2017remote}. In other words, the AoI is a perfect metric only when the model of the dynamic process under monitoring is not known by the receiver (see e.g.~\cite{kaul2012real}).
%	However, in practice, most of the dynamic processes are time-correlated, and the state-changing rules can be known by the receiver to some extent. 
%	Therefore, the receiver can estimate the current state of the process based on the previously received measurements and the model of the dynamic process (see e.g., \cite{schenato2008optimal,sun2017remote}), especially when the packet that carries the current sensor measurement is failed or delayed. 
%	In this sense, the estimation MSE is the perfect performance metric.

	From a communication protocol design perspective, we naturally ask: does a sensor need retransmission or not for mission-critical real-time remote estimation?
	Retransmission is required by conventional communication systems with non-real-time backlogged data to be perfectly delivered to the receivers.
	Also, energy-constrained remote estimation systems and the ones with low sampling rate can also benefit from retransmissions, see e.g.,~\cite{liu2016energy} and~\cite{Demirel2015to}. 
	It was shown in~\cite{gupta2010estimation} that retransmissions cannot improve the performance of a mission-critical real-time remote estimation system, which is not mainly constrained by energy nor sampling rate,
	as it is a waste of transmission opportunity to transmit an out-of-date measurement instead of the current one.
	However, this is true only when a retransmission has the same success probability as a new transmission, e.g., with the standard automatic repeat request (ARQ) protocol.
	Note that a hybrid ARQ (HARQ) protocol, e.g., with a chase combining (CC) or incremental redundancy (IR) scheme, is able to effectively increase the successful detection probability of a retransmission by combining multiple copies from previously failed transmissions~\cite{caire2001throughput}.
	Therefore, a HARQ protocol has the potential to improve the performance of real-time remote estimation.
%	However, to the best of our knowledge, HARQ has never been considered in the open literature of real-time remote estimation of time-correlated dynamic process.
	
	% which combines multiple copies from previous failed transmission to increase the successful decoding probability of retransmission, is commonly used to increase network robustness \cite{frenger2001performance,caire2001throughput}.

%	Although HARQ has been included in the 3G standard of wireless mobile telecommunications since 2007~\cite{seidel2006technology}, the  
%	recent communication standards for real-time wireless
%	control, such as WirelessHART~\cite{wirelessHART},
%	have not adopted any HARQ schemes.
%	Moreover, to the best of our knowledge, HARQ has never been considered in the open literature of real-time remote estimation of time-correlated dynamic process.
	In the paper, we introduce HARQ into real-time remote estimation systems and optimally design the sensor's transmission policy to minimize the estimation MSE.
	Note that there is a fundamental tradeoff between the reliability and freshness of the sensor's measurement transmission. When a failed transmission occurs, the sensor can either retransmit the previous old measurement such that the receiver can obtain a more reliable old measurement, or transmit a new but less reliable measurement.
	The main contributions of the paper are summarized as follows:
	\begin{itemize}
		\item We propose a novel HARQ-based real-time remote estimation system of time-correlated random processes, where the sensor makes online decision to send a new measurement or retransmit the previously failed~one depending on both the current estimation quality of the receiver and the current number of retransmissions of the sensor.
		\item We formulate the problem to optimize the sensor's transmission control policy so as to maximize the long-term performance of the receiver in terms of the average MSE for both the static and Markov fading channels.
		Since it is not clear whether the long-term average MSE can be bounded or not,
		we derive an elegant sufficient condition in terms of the transmission reliability provided by the HARQ protocol and parameters of the process of interest
		to ensure that an optimal policy exists and stabilizes the remote estimation system.
		\item We derive a structural property of the optimal policy, i.e., the optimal policy is a switching-type policy, and give an easy-to-compute suboptimal policy.
		Our numerical results show that the proposed HARQ-based optimal and suboptimal transmission control policies  significantly improve the system performance compared to the conventional non-HARQ policy, under the setting of practical system parameters. 
%		Also, we compare the MSE at the receiver of our performance-optimal policy with the delay-optimal (benchmark) policy that has been considered in the AoI literature. A noticeable performance gap between the policies highlights the superior of the proposed one.
		
	\end{itemize}
	
The remainder of this paper is organized as follows.
Section~\ref{section2} presents the proposed HARQ-based remote estimation system.
Section~\ref{sec:optimal transmission control}  analyzes the HARQ-based transmission-control policy, and formulate the optimal transmission control problem. 
Sections~\ref{sec:static channel} and~\ref{sec:Markov}  investigate the optimal transmission control policies of the static and Markov channels, respectively.
Section~\ref{sec:num} numerically presents the optimal, suboptimal and benchmark polices for both static and Markov channels, and their average MSE performance. 
Finally, Section~\ref{sec:con} concludes the paper.

\emph{Notations:} $\mathbb{P}\left[\mathcal{A}\right]$ is the probability of the event $\mathcal{A}$. $\mathbb{E}\left[A\right]$ is the expectation of the random variable $A$.
$(\cdot)^T$ is the matrix transpose operator. $\| \mathbf{v} \|_1$ is the sum of the vector $\mathbf{v}$'s elements. $\text{Tr}(\cdot)$ is the trace operator. $\text{diag}\{v_1,v_2,...,v_K\}$ denotes the diagonal matrix with the diagonal elements $\{v_1,v_2,...,v_K\}$. $\mathbb{N}$ and $\mathbb{N}_0$ denote the sets of positive and non-negative integers, respectively. $\left[u\right]_{B \times B}$ denotes the $B \times B$ matrix with identical element~$u$.

\section{System Model} \label{section2}
We consider a basic system setting that a \emph{smart sensor} periodically samples, pre-estimates and sends its local estimation of a dynamic process to a remote receiver through a wireless link with packet dropouts, as illustrated in Fig.~\ref{retransmit_system_model}.
     \begin{figure}[t]
	\centering\includegraphics[scale=0.7]{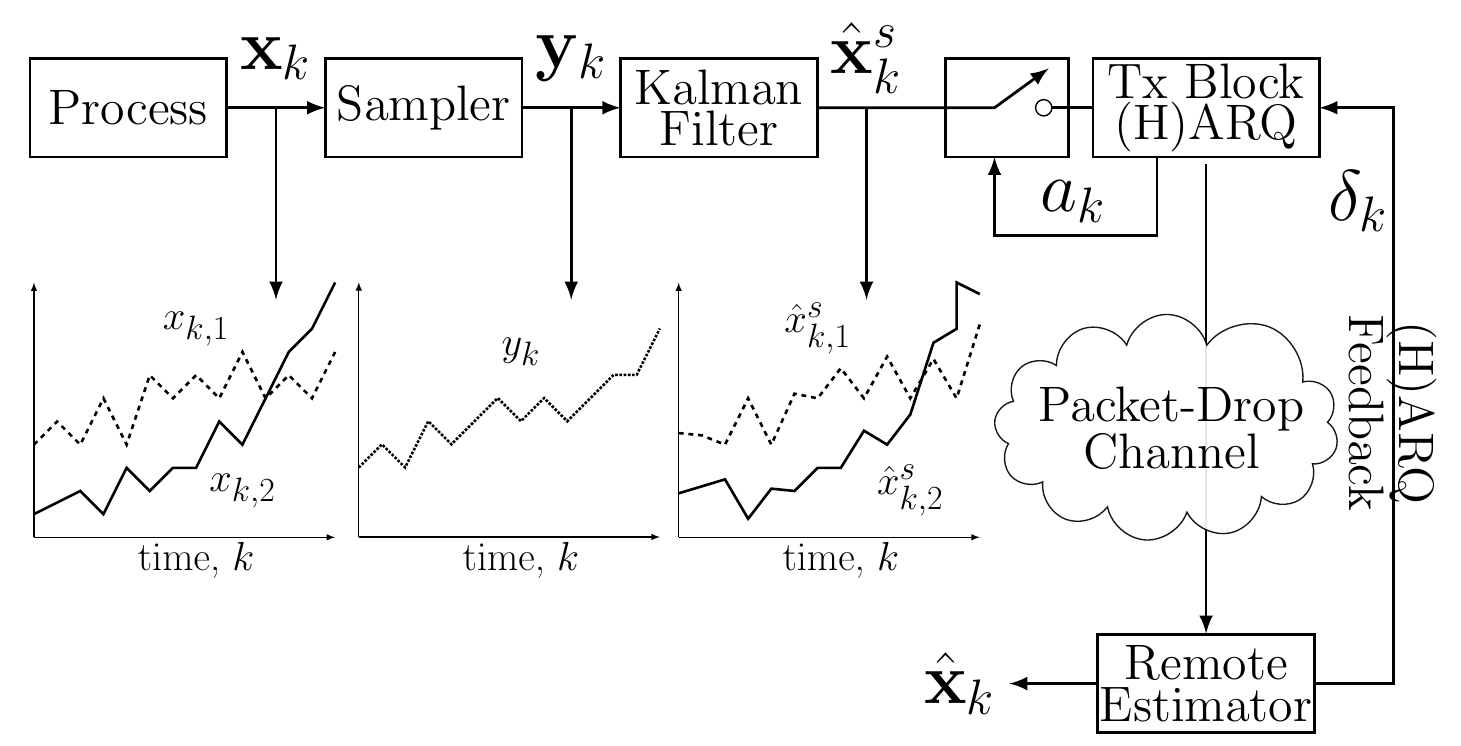}
	\vspace{-0.5cm}
	\caption{Proposed remote estimation system with HARQ, where $\mathbf{x}_k\triangleq \left[x_{k,1},x_{k,2}\right]^T$ is the two-dimensional state vector of the dynamic process, and $\hat{\mathbf{x}}^s \triangleq \left[\hat{x}^s_{k,1},\hat{x}^s_{k,1}\right]^T$.}
	\label{retransmit_system_model}
	\vspace{-0.5cm}
\end{figure}
\subsection{Dynamic Process Modeling}
   We consider a general discrete LTI model for the dynamic process as (see e.g., \cite{schenato2008optimal,shi2012optimal,yang2013schedule})
   \begin{equation} \label{sys}
	\begin{aligned}
	   \mathbf{x}_{k+1} &= \mathbf{A} \mathbf{x}_k + \mathbf{w}_k,\\
	   \mathbf{y}_k &= \mathbf{C}\mathbf{x}_k + \mathbf{v}_k,
	\end{aligned}
   \end{equation}
     where 
     the discrete time steps are determined by the sensor's sampling period $T_s$,
     $\mathbf{x}_k \in \mathbb{R}^n$ is the process state vector, $\mathbf{A} \in \mathbb{R}^{n \times n}$ is the state transition matrix, $\mathbf{x}_k \in \mathbb{R}^m$ is the measurement vector of the smart sensor attached to the process, $\mathbf{C} \in \mathbb{R}^{m \times n}$ is the measurement matrix\footnote{Note that $C$ is not necessary to be full rank~\cite{maybeck1979stochastic}, as illustrated in Fig.~\ref{sys}, i.e., $\mathbf{x}_k$ is a two-dimensional (2D) signal, while the measurement $y_k$ is one-dimensional. After Kalman filtering, we have a 2D $\hat{\mathbf{x}}^s_k$.}, $\mathbf{w}_k \in \mathbb{R}^n$ and $\mathbf{v}_k \in \mathbb{R}^m$ are the process and measurement noise vectors, respectively. We assume $\mathbf{w}_k$ and $\mathbf{v}_k$ are independent and are identically distributed (i.i.d.) zero-mean Gaussian processes with corresponding covariance matrices $\mathbf{Q}_w$ and $\mathbf{Q}_v$, respectively. The initial state $\mathbf{x}_0$ is zero-mean Gaussian with covariance matrix $\mathbf{\Sigma}_0$.     
     To avoid trivial problems, we assume that $\rho^2(\mathbf{A}) > 1$, where $\rho^2(\mathbf{A})$ is the maximum squared eigenvalue of $\mathbf{A}$~\cite{shi2012scheduling}. 
     %otherwise, the system state $x_k$ converges to the i.i.d. zero-mean process $w_k$ when $k$ goes to infinity.

\subsection{State Estimation at the Smart Sensor}
Since the sensor's measurements are noisy, the smart sensor with sufficient computation and storage capacity is required to estimate the state of the process, $\mathbf{x}_k$, using a Kalman filter~\cite{shi2012optimal,yang2013schedule}, which gives the minimum estimation MSE, based on the current and previous raw measurements:
%
%
% When the packets drop link is considered, transmitting local state estimation instead of the raw measurements gives better performance in receiver \cite{gupta2005lqg}.
%
%     
%Sensor attached to the process has computation capacity and sufficient storage to run a local Kalman filter, which gives the optimal mean square error (MSE) estimation of $x_k$~\cite{maybeck1979stochastic}:
\begin{subequations}\label{sub:1}
\begin{align}
\mathbf{x}_{k|k-1}^s&=\mathbf{A} \mathbf{x}_{k-1|k-1}^s\\
\mathbf{P}_{k|k-1}^s&=\mathbf{A} \mathbf{P}_{k-1|k-1}^s \mathbf{A}^T+\mathbf{Q}_w\\
\mathbf{K}_k&=\mathbf{P}_{k|k-1}^s \mathbf{C}^T(\mathbf{C} \mathbf{P}_{k|k-1}^s \mathbf{C}^T+\mathbf{Q}_v)^{-1}\\
\mathbf{x}_{k|k}^s&=\mathbf{x}_{k|k-1}^s+\mathbf{K}_k(\mathbf{y}_{k}-\mathbf{C} \mathbf{x}_{k|k-1}^s)\\
\mathbf{P}_{k|k}^s&=(\mathbf{I}-\mathbf{K}_{k} \mathbf{C})\mathbf{P}_{k|k-1}^s
\end{align}
\end{subequations}
where $\mathbf{I}$ is the $m \times m$ identity matrix, $\mathbf{x}^s_{k|k-1}$ is the priori state estimation, $\mathbf{x}^s_{k|k}$ is the posteriori state estimation at time $k$, $\mathbf{K}_k$ is the Kalman gain, $\mathbf{P}_{k|k-1}$ and $\mathbf{P}_{k|k}$ represent the priori and posterior error covariance at time $k$, respectively. The first two equations present the prediction steps while the last three equations present the updating steps~\cite{maybeck1979stochastic}.
Note that $\mathbf{x}^s_{k|k}$ is the output of the Kalman filter at time $k$, i.e., the pre-filtered measurement of~$\mathbf{y}_k$, with the estimation error covariance $\mathbf{P}_{k|k}^s$.

As we focus on the effect of communication protocols on the stability and quality of the remote estimation, we assume that the local estimation is stable as follows~\cite{shi2012optimal,yang2013schedule}.

\begin{assumption}
	\normalfont
	The local Kalman filter of system \eqref{sys} is stable with the system parameters $\{\mathbf{A}, \mathbf{C}, \mathbf{Q}_w\}$\footnote{The rigorous stability condition in terms of $\{\mathbf{A}, \mathbf{C}, \mathbf{Q}_w\}$ is given in~\cite{maybeck1979stochastic}.}, i.e.,
	the error covariance matrix $\mathbf{P}_{k|k}^s$ converges to a finite matrix~$\bar{\mathbf{P}}_0$ when $k$ is sufficiently large.
\end{assumption}

\emph{In the rest of the paper, we assume that the local Kalman filter operates in the steady state~\cite{shi2012optimal,yang2013schedule}, i.e., $\mathbf{P}_{k|k}^s = \bar{\mathbf{P}}_0$. To simplify the notation, we use $\hat{\mathbf{x}}_k^s$ to denote the sensor's estimation, $\mathbf{x}_{k|k}^s$.}
%\begin{figure}[ht]
%	\centering\includegraphics[scale=0.5]{kalman_filter.pdf}
%	\caption{Example for kalman filter, where there are two state variables and one measurement output and $A = [1.1,0.2;0.2,0.8],C = [1,1], Q = I, R = 1$. }
%	\label{kalman filter}
%\end{figure}
       
\subsection{Wireless Channel}   \label{sec:ARQ} 
We consider both a static channel and a finite-state time-homogeneous Markov fading channel.
For the static channel, the channel power gain $h_k$ does not change with time, i.e., $h_k=h>0$, $\forall k$.
For the Markov channel, the channel power gain $h_k$ remains constant during the $k$th time slot and changes slot by slot, where $h_k>0, \forall k$. We assume that the Markov channel has $B$ states, i.e, $\mathcal{U}\triangleq\{u_1,...,u_B\}$, and $h_k \in \mathcal{U}$.
The probability of transition from state $i$ to state $j$ is $p_{i,j}$,
and the matrix of channel state transition probability is given as
\begin{equation} \label{P_matrx}
\mathbf{\Pi} \triangleq \begin{bmatrix}
p_{1,1}  & \cdots &	p_{B,1}\\
\vdots  & \ddots &	\vdots\\
p_{1,B}  & \cdots &	p_{B,B}
\end{bmatrix}.
\end{equation}

We assume that the channel state information is available at both the sensor and the receiver, see e.g.~\cite{markovchannel} and the references therein.

\subsection{HARQ-Based Communication}  
The sensor's estimation is quantized into $(L \times R)$ bits and then coded into a packet with $L$ symbols, where the symbol duration is $T'_s$ and $R$ is the coding rate. 
We assume that the packet length is equal to the sampling period, i.e., $L T'_s = T_s$.
In other words, the sensors perform the next sampling once the current measurement-carrying packet has been delivered to the receiver.
Thus, there exists a unit packet-transmission delay between the sensor and the receiver.
For example, the sensor's measurement at the beginning of time slot~$k$, ${\mathbf{y}}_k$, is filtered and sent to the receiver before time slot $(k+1)$. It is assumed that the sensor and the receiver are perfectly synchronized.

The acknowledgment/negative-acknowledgment
(ACK/NACK) message is fed back from the receiver to the sensor perfectly without any delay, when the packet detection succeeds/fails.
If an ACK is received by the sensor, it will send a new (pre-filtered) measurement in the next time slot. If a NACK is received, the sensor can decide whether to retransmit the unsuccessfully transmitted measurement using its ARQ protocol or to send a new measurement.
We introduce the binary variable $\gamma_k \in \{1,0\}$ to indicate the successful and failed packet detection in time slot $k$, respectively.

For the standard ARQ protocol, the receiver discards the failed packets, and the sensor simply resends the previously failed packet if a retransmission is required.
Thus, the successful packet detection probability~at each time is independent of the current number of retransmissions.

For a HARQ protocol, the receiver buffers the incorrectly received packets,
and the detection of the retransmitted packet will utilize all the buffered related packets.
In the CC-HARQ case, the sensor resends the previously failed packet if a retransmission is required, and the receiver optimally combines (i.e., the maximal ratio combining method) 
all the previously received replicas of the packet of the same message and make a detection.
In the IR-HARQ case, each retransmitted packet is an incremental redundancy of the same message, and the receiver treats the sequence of all the buffered replicas as a long codeword to detect the transmitted massage.

Given the channel power gains, the probability that the message cannot be detected within $r$ transmission attempts started from time slot $k-(l-1)$ is given as~\cite{Polyanskiy} and \cite{li2017throughput}
\begin{align}\label{HARQ_error}
&\mathbb{P}\left[\xi_k = 0 \vert l \text{ transmission attampts} \right] \\
%&\triangleq
%\tilde{g}\left(h_k, h_{k+1}, ...,h_{k+r-1}\right)\\ 
\label{approx_HARQ}
&\approx
\begin{cases}
Q \bigg(\frac{{L^{\frac{1}{2}}}(\log_{2} (1+\sum_{i=0}^{l-1} h_{k-i} \mathsf{SNR})+{\frac{\log_{2} L}{L}}-R)}{\sqrt{1-\frac{1}{(1+\sum_{i=0}^{l-1} h_{k-i} \mathsf{SNR})^2}}\log_{2} e}\bigg), &\text{CC-HARQ}\\
Q \bigg(\frac{L^{\frac{1}{2}}(\sum_{i=0}^{l-1}\log_{2} (1+h_{k-i} \mathsf{SNR})+{\frac{\log_{2} rL}{L}}-R)}{\sqrt{\sum_{i=0}^{l-1}(1-\frac{1}{(1+h_{k-i} \mathsf{SNR})^2})}\log_{2} e}\bigg), &\text{IR-HARQ}
\end{cases}
\end{align}
where $\xi_k = 0$ is the event of failed detection within $r$ transmissions in time slot $k$, $\mathsf{SNR}$ is the signal-to-noise ratio at the receiver with unit channel power gain, and the approximation \eqref{approx_HARQ} is based on the results of the finite-blocklength information theory for AWGN channel, see e.g.,~\cite{Polyanskiy} and \cite{li2017throughput}.

\subsection{State Estimation at the Receiver}\label{sec:estimation}
We assume that 
the latest sensor's estimation that is available at the receiver at the beginning of time slot $k$, i.e., $\hat{\mathbf{x}}^s_{t_k}$, was generated at the beginning of time slot $t_k$.
Therefore, the receiver-side AoI at the beginning of time slot $k$ can be defined as~\cite{kaul2012real}
\begin{equation}\label{def_q}
q_k \triangleq k-t_k,\forall k,
\end{equation}
and $q_k\geq 1$.

As the latest available sensor's estimation was generated $q_k$-step earlier, the receiver needs to estimate the current state based on the dynamic process model~\eqref{sys}. 
The receiver's MSE optimal estimator at the beginning of time slot $k$ is given~as~\cite{schenato2008optimal} 
\begin{equation} \label{general_estimater}
\hat{\mathbf{x}}_k = 
\mathbf{A}^{k-t_k}\hat{\mathbf{x}}_{t_k}^s
=\mathbf{A}^{q_k}\hat{\mathbf{x}}_{t_k}^s,
\end{equation}
and the corresponding estimation error covariance is
\begin{align} \label{covariance1}
\mathbf{P}_k 
&\triangleq  \mathbb{E}\left[(\mathbf{x}_k-\hat{\mathbf{x}}_k)(\mathbf{x}_k-\hat{\mathbf{x}}_k)^T\right]\\
&=f^{q_{k}}(\bar{\mathbf{P}}_0) \label{general_form}
\end{align}
where \eqref{general_form} is obtained by substituting \eqref{general_estimater} and \eqref{sys}  into \eqref{covariance1},
$f(\mathbf{X})\triangleq \mathbf{AXA}^T+\mathbf{Q}$, $f^{n+1}(\cdot)  \triangleq f (f^{n}(\cdot))$ when $n\geq 1$, and $f^{1}(\cdot) \triangleq f(\cdot)$.	
Note that $\mathbf{P}_k$ takes value from a countable infinity set~\cite{shi2012scheduling}, i.e., $\mathbf{P}_k \in \{f(\bar{\mathbf{P}}_0),f^2(\bar{\mathbf{P}}_0),\cdots\}$.

The receiver's estimation MSE at the beginning of time slot $k$ is $\text{Tr}\left(\mathbf{P}_k\right) $.
Note that the operator $\text{Tr}\left(f^n(\cdot)\right)$ is monotonically increasing with respect to (w.r.t.) $n$, i.e.,  $\text{Tr}\left(f^{n_1}(\bar{\mathbf{P}}_0)\right) \leq \text{Tr}\left(f^{n_2}(\bar{\mathbf{P}}_0)\right)$ if $\rho^2(\mathbf{A})>1$ and $1 \leq n_1 \leq n_2$ (see Lemma 3.1 in \cite{shi2012scheduling}).

\begin{remark}
From \eqref{general_form}, the estimation MSE is a non-linear function of the AoI, and thus, $q_k$ can also be treated as the estimation quality indicator
\end{remark}

\subsection{Performance Metric}
The long-term average MSE of the dynamic process is defined as
\begin{equation}\label{cost_function}
\limsup_{K\to\infty}\frac{1}{K}\sum_{k=1}^{K} \mathbb{E}\left[\text{Tr}\left(\mathbf{P}_k\right)\right],
\end{equation}
where $\limsup_{K\rightarrow \infty}$ is the limit superior operator.

\section{Optimal Transmission Control: Analysis and Problem Formulation}\label{sec:optimal transmission control}
For the standard ARQ,
as the chance of the successful detection of a new transmission and that of a retransmission are the same, 
\emph{the optimal policy is to always transmit the current sensor's estimation, i.e., a non-retransmission policy~\cite{gupta2010estimation}.}
For a HARQ protocol, the probability of successful packet detection in time slot $k$, depends on the number of consecutive transmission attempts of the original message and the experienced channel conditions.
Since a new transmission is less reliable than a retransmission, there exists an inherent trade-off between retransmitting previously failed local state estimation with a higher success probability, and sending the current state estimation with a lower success probability.
Therefore, when a packet detection error occurs, the sensor needs to optimally make a decision on whether to retransmit it or to start a new transmission.
%depends on the number of consecutive transmission attempts $r_k \geq 1$ of the original message and the experienced channel conditions. In particular, $r_k = 0$ indicates a new transmission in time slot~$k$.
%Since a new transmission is less reliable than a retransmission, there exists an inherent trade-off between retransmitting previously failed local state estimation with a higher success probability, and sending the current state estimation with a lower success probability.
%Therefore, when a detection error occurs, the sensor has to properly make online decision whether to retransmit it or to start a new estimation.

\subsection{Transmission-Control Policy}
Let $a_k \in \{0,1\}$ be the sensor's decision variable at time~$k$ as illustrated in Fig.~\ref{sys}. If $a_k=0$, the sensor sends the new measurement to the receiver in time slot $k$; otherwise, it retransmits the  unsuccessfully transmitted measurement. It is clear that $a_k=0$, if the the packet transmitted in time slot $(k-1)$ was successful. 

Let $r_k$ denote the number of consecutive transmission attempts before time slot $k$.
As $r_k$ only depends on the sensor's transmission-control policy, it has the updating rule~as
\begin{equation} \label{retransmission time}
r_k =\begin{cases}
1, &\mbox{if } a_{k-1} = 0 \\
r_{k-1}+1, &\mbox{otherwise}, 
\end{cases}
\end{equation}
where $r_k \geq 1, \forall k$.
From the definition of the estimation quality indicator~\eqref{def_q}, the updating rule of $q_k$ is given as
\begin{equation} \label{q_envolution_1}
q_{k} = 
\begin{cases}
r_{k}, & \gamma_{k-1} =1\\
q_{k-1}+1, & \gamma_{k-1} = 0,
\end{cases}
\end{equation}
where $q_k \geq 1, \forall k$.
As the estimation quality indicator depends on the current number of transmission attempts and also the control policy, plugging \eqref{retransmission time} into \eqref{q_envolution_1}, we further have
\begin{equation} \label{q_envolution_2}
q_{k} = 
\begin{cases}
1, & a_{k-1} = 0, \gamma_{k-1} =1\\
r_{k-1}+1, & a_{k-1} = 1, \gamma_{k-1} = 1\\
q_{k-1}+1, & \gamma_{k-1} = 0.
\end{cases}
\end{equation}

\subsection{Packet Error Probability with Online Transmission Control}
If the sensor decides to transmit a new measurement in time slot $k$, i.e., $a_k=0$, the packet error probability in time slot $k$ is obtained directly from \eqref{HARQ_error} as
\begin{align}\label{error_1}
\mathbb{P}\left[\gamma_k = 0 \vert a_k=0 \right] 
&=\mathbb{P}\left[\xi_k = 0 \vert l = 1\right].
\end{align}
If a retransmission decision has been made, i.e., $a_k =1$, the packet error probability based on \eqref{HARQ_error} can be obtained as
\begin{align} 
\mathbb{P}\left[\gamma_k = 0 \vert a_k=1\right] 
&=\mathbb{P}\left[\gamma_k = 0 \vert \gamma_{k-1},...,\gamma_{k-r_k}= 0 \right] \\
&=\frac{\mathbb{P}\left[\gamma_k, \gamma_{k-1},...,\gamma_{k-r_k}= 0  \right] }{\mathbb{P}\left[\gamma_{k-1},...,\gamma_{k-r_k}= 0  \right] }\\\label{error_2}
&=\frac{\mathbb{P}\left[\xi_k = 0 \vert l = r_k+1 \right] }{\mathbb{P}\left[\xi_{k-1} = 0 \vert l = r_k \right] }.
\end{align}

\underline{In the Markov channel scenario}, we assume that the packet error probability in \eqref{error_2} is a function of the current channel power gain $h_k$, and the state indicator $\mathbf{\Omega}_k$ of the previously experienced $r_k$ channel states, which does not rely on the order of the channel states\footnote{This assumption is in line with the approximation in \eqref{approx_HARQ}.}.
To be specific, we define
$\mathbf{\Omega}_k \triangleq [r_1(k),r_2(k),...,r_B(k)]$, where $r_i(k)$ is the occurrence number of the channel state with channel power gain $u_i$, during time slots $k-r_k$ to $k-1$. 
In other words, $\mathbf{\mathbf{\Omega}}_k$ is a sorted counter of the relevant historical channel states, and $\mathbf{\Omega}_k\in \mathbb{N}^B_0\backslash\mathbf{0}$.
By introducing the channel state index $\Xi_k \in \{1,\cdots,B\}$, i.e., $h_k =u_{\Xi_k} \in \mathcal{U}$,
$\mathbf{\Omega}_k$ is updated as
\begin{equation} \label{Markov_channel_evolution}
\mathbf{\Omega}_k=\begin{cases}
\mathbf{1}_{\Xi_{k-1}}, & \mbox{if } a_{k-1} = 0\\
\mathbf{\Omega}_{k-1}+ \mathbf{1}_{\Xi_{k-1}}, & \mbox{if } a_{k-1} = 1
\end{cases}
\end{equation}
where $\mathbf{1}_i \triangleq [\underbrace{\overbrace{0,\cdots,0,1}^{i},0,\cdots,0}_{B}]$.

Therefore, the packet error probability in \eqref{error_1} and \eqref{error_2} can be uniformly written as:
\begin{equation} \label{Markov_channel_loss}
\mathbb{P}\left[\gamma_k = 0  \right] 
=
\left\lbrace \begin{aligned}
& \tilde{g}(\mathbf{0},\Xi_{k})
 \triangleq \mathbb{P}\left[\xi_k = 0 \vert r = 1\right], &&a_k = 0, \\
& \tilde{g}(\mathbf{\Omega}_k,\Xi_{k})
\triangleq \frac{\mathbb{P}\left[\xi_k = 0 \vert r = r_k+1 \right] }{\mathbb{P}\left[\xi_{k-1} = 0 \vert r = r_k \right] }, &&a_k = 1.
\end{aligned}
\right.
\end{equation}
where $\mathbf{0} \triangleq [\underbrace{0,0,...,0}_{B}]$.

\underline{In the static channel scenario}, i.e., a special case of the Markov channel scenario, as the channel power gains are identical to each other, the packet error probability in \eqref{Markov_channel_loss} can be rewritten as a function of the current number of transmission attempts as
\begin{equation} \label{static_channel_loss}
\mathbb{P}\left[\gamma_k = 0  \right] 
=
\left\lbrace \begin{aligned}
& g(1), &&a_k = 0, \\
& g(r_k+1), &&a_k = 1.
\end{aligned}
\right.
\end{equation}

As the packet error probability of a retransmission is smaller than a new transmission under the same channel condition, we have the following inequalities for the Markov and static channel scenarios, respectively, as
\begin{align} \label{inequality_1}
\Lambda'_i \triangleq \tilde{g}(\mathbf{0},i) > \tilde{g}(\mathbf{\Omega}_k,i), \forall k,
\end{align}
and 
\begin{align} \label{inequality_2}
\Lambda'_0 \triangleq g(1) > g(r_k+1), \forall k,
\end{align}
where $\Lambda'_i$ is the packet error probability of a new transmission with the channel power gain $u_i$ in the Markov channel scenario, and $\Lambda'_0$ is the packet error probability of a new transmission in the static channel scenario.

For the Markov channel, the largest packet error rate of a retransmission with channel power gain $u_i$ is defined as
\begin{equation}\label{max_error_markov}
\Lambda_i \triangleq \max\limits_{\mathbf{\Omega} \in \mathbb{N}^B_0 \backslash 
	 \mathbf{0}} \tilde{g}(\mathbf{\Omega},i), i=1,2,...,B.
\end{equation}
For the static channel, the largest packet error rate of a retransmission is defined as
\begin{equation} \label{max_error_static}
\Lambda_0  \triangleq \max\limits_{r>1} g(r).
\end{equation}

\subsection{Problem Formulation}
The sensor's transmission control policy is defined as the sequence
$\left\lbrace a_1,a_2,...,a_k,\cdots \right\rbrace$, where $a_k$ is the control action in time slot $k$.
In what follows, we optimize the sensor's  policy such that the long-term estimation error is minimized,~i.e.,
\begin{equation} \label{problem}
\min_{ a_1,a_2,...,a_k,\cdots }\limsup_{K\to\infty}\frac{1}{K}\sum_{k=1}^{K} \mathbb{E}\left[\text{Tr}\left(\mathbf{P}_k\right)\right].
\end{equation}

It is possible that the long-term estimation error may never be bounded no matter how we choose the policy, if the channel quality is always bad or the dynamic process~\eqref{sys} changes rapidly.
Therefore, \emph{it is also important to investigate the condition in terms of the transmission reliability and the dynamic process parameters, under which the remote estimation system can be stabilized, i.e., the long-term estimation MSE can be bounded.}

To shed light on the stability condition and the optimal policy structure, we first consider the simplified case, i.e., the static channel scenario, in Sec.~\ref{sec:static channel}.
The insights are leveraged to investigate the general Markov channel scenario in Sec.~\ref{sec:Markov}.

% \begin{figure}[t]
%	\centering\includegraphics[scale=1]{arrow.pdf}
%	\vspace{-0.5cm}
%	\caption{An illustration of the sensor's transmission process. The solid circles denote the raw measurement sampling time, 
%		the up arrows are the starting points of new transmissions (i.e., only these (pre-filtered) measurements will be sent to the receiver), solid/dashed blocks are new/re-transmission packets, and $\checkmark$/$\times$ denotes a successful/failed detection at the receiver.}
%	\label{fig:packet_process}
%	\vspace{-0.5cm}
%\end{figure} 

\section{Optimal Policy: Static Channel} \label{sec:static channel}
In this section, we investigate the optimal transmission control policy in static channel.

\subsection{MDP Formulation} \label{sec:MDP}  
From \eqref{general_form}, \eqref{retransmission time} and \eqref{q_envolution_2}, the estimation MSE $\text{Tr}\left(\mathbf{P}_k\right)$ and states $r_k$ and $q_k$ only depend on the previous action and states, i.e., $a_{k-1}$, $r_{k-1}$ and $q_{k-1}$.
Therefore, the online decision problem \eqref{problem} can be formulated as a discrete time Markov decision process (MDP) as follows. 
   %The quadruplet $(\mathbb{S},\mathbb, P(\cdot|\cdot), c(\cdot | \cdot))$ is elaborated as follows: 

1) The state space is defined as 
$\mathbb{S} \triangleq \{(r, q) : r \leq q,\ (r, q) \in \mathbb{N} \times \mathbb{N}\}$, where the number of transmission attempts, $r$, should be no larger than the estimation quality indicator (i.e., the AoI), $q$, from the definition.
The state of the MDP at time $k$ is 
$s_k \triangleq (r_k, q_k) \in \mathbb{S}$. 
%
%$q_k=0$ means the current remote estimation error covariance is equal to $f(\bar{P}_0)$. At time $k+1$, if sensor chooses to retransmit previous estimation and it is successfully received by receiver, then $q_{k+1}=r_k + 1$. 

2) The action space is defined as $\mathbb{A} \triangleq \{0,1\}$.
A policy is a mapping from states to actions, i.e., $\pi: \mathbb{S}\rightarrow \mathbb{A}$.
Recall that the action at time $k$, $a_k \triangleq \pi(s_k) \in \mathbb{A}$, indicates a new transmission $(a_k=0)$ or a retransmission $(a_k=1)$.

3) The state transition function $P(s'|s,a)$ characterizes the probability that the state transits from state $s$ at time $(k-1)$ to $s'$ at time $k$ with action $a$ at time $k-1$. As the transition is time-homogeneous, we can drop the time index $k$ here. Let $s=(r,q)$ and $s'=(r',q')$ denote the current and next state, respectively. Based on the packet error probability \eqref{static_channel_loss} and the state updating rules \eqref{retransmission time} and \eqref{q_envolution_2}, we have the state transition function as
%
%If the action $a=0$, the next state is 
%\begin{equation}\label{s'_1}
%s'=
%\begin{cases}
%(1,1), \mbox{ with probability } (1-g(1)) \\
%(1,q+1), \mbox{ with probability } g(1).
%\end{cases}
%\end{equation}
%If the action $a=1$, the next state is 
%\begin{equation}\label{s'_2}
%\hspace{-0.2cm}
%s'=\begin{cases}
%(r+1,r+1), \mbox{ with probability } (1-g(r+1))\\
%(r+1, q+1), \mbox{ with probability } g(r+1).
%\end{cases}
%\end{equation}
\begin{equation} \label{static:transision_func}
P(s'\vert s, a)=
\left\lbrace
\begin{aligned}
& 1-g(1), &&\text{if }a=0,s'= (1,1) \\
& g(1), &&\text{if }a=0,s'= (1,q+1) \\
& 1-g(r+1), &&\text{if }a=1,s'= (r+1,r+1) \\
& g(r+1), &&\text{if }a=1,s'= (r+1,q+1) \\
& 0, &&\text{otherwise}.\\
\end{aligned}
\right.
\end{equation}

%From (9) and (10), we conclude that if current state $s=(r,q)$ satisfies $r\leq q$, then the next state $s^{'}=(r^{'},q^{'})$ must satisfy $r^{'} \leq q^{'}$. By setting our initial state at (0,0), we can reduce the state space of the MDP to $s \triangleq (r, q) \in \mathbb{N}^2, r \leq q$.

4) The one-stage (instantaneous) cost, i.e., the estimation MSE based on \eqref{general_form}, is a function of the current state of $q$:
\begin{equation} \label{one-stage cost}
c(s,a) \triangleq
\text{Tr}\left(f^{q}(\bar{\mathbf{P}}_0)\right),
\end{equation}
which is independent of the action.

Therefore, the problem \eqref{problem} is equivalent to solve the classical \emph{average cost optimization} problem of the MDP.
Assuming the existence of a stationary and deterministic optimal policy, we can effectively solve the MDP problem using standard methods such as the relative value iteration algorithm~\cite[Chapter~8]{puterman2014markov}.

\subsection{Optimal Policy: Condition of Existence}
Since the cost function grows exponentially with the state~$q$, it is possible that the long-term average cost with a HARQ-based transmission control policy, $\pi$, in the state space $\mathbb{S}$ cannot be bounded, i.e., the remote estimation system is unstable. We give the following sufficient condition of the existence of an optimal policy that has a bounded long-term estimation MSE.

%Note that the non-retransmission (standard ARQ-based) policy is a special case of the one with HARQ, where the action $a$ always equals to zero, i.e., a deterministic policy.
%Also, it is well-known that this policy is stationary, and the average long-term MSE is bounded if and only if the following condition satisfies~\cite{schenato2008optimal}:

%Since a stationary and deterministic optimal policy exists if there is a stationary and deterministic policy with a bounded long-term cost~\cite{sennott1996convergence}, we have the following result.\footnote{A more rigorous proof will be given in our journal version.}

 \begin{theorem} \label{theorem:existence}
 	\normalfont
    	For the static channel, there exists a stationary and deterministic optimal policy $\pi^*$ of problem \eqref{problem}, if the following condition holds:
    	\begin{equation} \label{stability_condition}
    	\Lambda_0 \rho^2(\mathbf{A}) <1,
    	\end{equation}
    	where $\Lambda_0$ is the largest packet error probability of a retransmission defined in \eqref{max_error_static}.
 \end{theorem}
\begin{proof}
	See Appendix A.
\end{proof}

% \begin{proof}[Sketch Proof]
% 	The sketch proof is given as here. First we consider the policy $\pi'$ that always transmit current state estimation with packet reception probability $g(0)$. Note that if we assume the process starts at $r_0=0$ and $q_0 = 0$, which is reasonable, then this scenario is the same as the situation studied in \cite{schenato2008optimal}. The stability condition is given as either $A$ is stable or $\rho^2(A) \cdot (1-\lambda)<1$. Since $\pi'$ is one of all feasible policies $\Pi$, the covariance of the optimal policy is always smaller than the covariance of $\pi'$, which will be bounded consequently. Therefore, the cost of each stage of optimal policy is also bounded, which will give a deterministic policy that minimize problem (8). As the transition matrix and reward for each state are independent of time k, the optimal policy is also stationary. 
% \end{proof}

\begin{remark}
	From Theorem~\ref{theorem:existence}, it is clear that the optimal policy exists if we have a good channel condition and a good HARQ scheme, which guarantees high retransmission reliability (i.e., a small $g(r)$ and hence a small $\Lambda_0$), or the dynamic process does not change quickly, which is easy to estimate (i.e., a small $\rho^2(\mathbf{A})$).	
\end{remark}

%\begin{proposition}[\cite{schenato2008optimal}]
%
%	is defined in \eqref{inequality_2}.
%\end{proposition}

\subsection{Optimal Policy: The Structure}
We show that the optimal policy has a switching structure as follows.
	\begin{theorem} \label{theorem:switching}
 	\normalfont		
		The optimal policy $\pi^*$ of problem \eqref{problem} is a switching-type policy, i.e.,		
(i) if $\pi^{*}(r,q)=0$, then  $\pi^{*}(r+z,q)=0$;
(ii) if $\pi^{*}(r,q)=1$, then  $\pi^{*}(r,q+z)=1$, where $z$ is any positive integer.
	\end{theorem}
	\begin{proof}
		See Appendix B.
	\end{proof}
In other words, for the optimal policy, the two-dimensional state space $\mathbb{S}$ is divided into two regions by a curve, and the decision actions of the states within each region are the same, as illustrated in Fig.~\ref{fig:region_action}.
%\begin{figure}[t]
%	\centering
%	\includegraphics[scale=1]{region_action.pdf}
%	\vspace{-0.3cm}
%	\caption{The switching structure of the optimal policy in the state space $\mathbb{S}$.}
%	\label{fig:region_action}
%	\vspace{-0.5cm}
%\end{figure}
\begin{figure*}[t]
	%	\renewcommand{\captionlabeldelim}{ }	
	%	\renewcommand{\captionfont}{\small} \renewcommand{\captionlabelfont}{\small}
	%		\hspace{-15pt}
	\minipage{0.5\textwidth}
	\centering
	%	\vspace*{-0.9cm}	
	\includegraphics[scale=1]{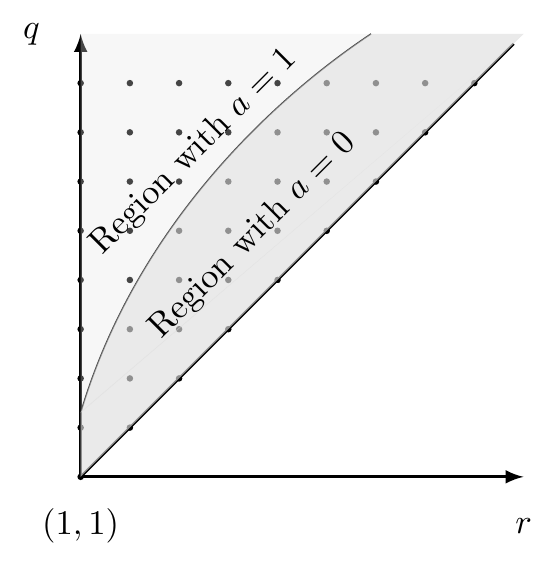}
	\vspace{-0.3cm}
	\caption{The switching structure of the optimal policy in the state space $\mathbb{S}$.}
	\label{fig:region_action}
	%pmfs for $\Tstd$ and $\Tuc$ with
	\endminipage
	\hspace{0.5cm}
	\minipage{0.45\textwidth}	
	\centering
%	\vspace*{-0.0cm}
	\includegraphics[scale=1]{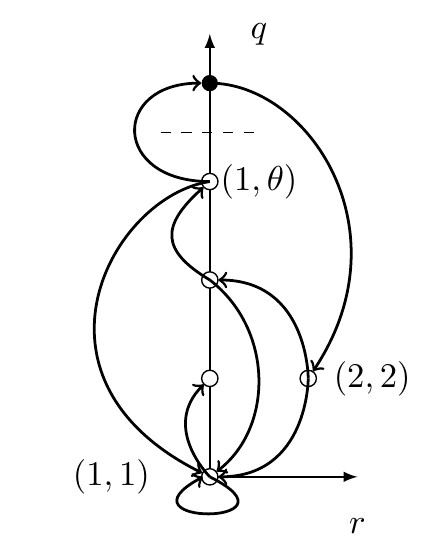}
	\vspace{0.0cm}
	\caption{Illustration of the state space and state transitions in the high SNR scenario, where $\theta = 4$.}
	\label{fig:special_static}
	\endminipage
	%	\hrulefill
	\vspace*{-0.7cm}
\end{figure*}

\begin{remark}
%	[to be revised] We give the intuitive explanation of the switching structure. As the number of retransmission, $r$, increases, the retransmitted packets become more "inaccuracy" and it is better for transmit current state estimation. As the covariance grows, the "cost" of transmission loss grows as well, which drive sensor to choose more reliable transmission (retransmission).	
	Note that the switching structure can help save storage space for on-line implementation, since the smart sensor only needs to store switching-boundary states rather than the actions on the entire state space. At each time, the sensor simply needs to compare the current state with the boundary states to give the optimal decision.
\end{remark}
     
\subsection{Optimal Policy: A Special Case}
We consider the high SNR scenario, where retransmissions can have ultra-low packet error probabilities. Therefore, we assume that a retransmission is always successful in the high SNR scenario, and the optimal policy always exists from Theorem~\ref{theorem:existence}.

%\begin{figure}[t]
%	\centering
%	\includegraphics[scale=0.65]{HARQ_error_prob.pdf}
%	\vspace{-0.3cm}
%	\caption{The packet error rates with different number of transmission attempts, $\mathsf{SNR}=12$ dB, $L=100$ and $R=4$.}
%	\label{fig:HARQ_error_prob}
%	\vspace{-0.5cm}
%\end{figure}
%
%\begin{figure}[t]
%	\centering
%	\includegraphics[scale=1]{special_static.pdf}
%	\vspace{-0.3cm}
%	\caption{Illustration of the state space and state transitions in the high SNR scenario, where $\theta = 4$.}
%	\label{fig:special_static}
%	\vspace{-0.5cm}
%\end{figure}

Due to the successful retransmissions, it can be noted from \eqref{static:transision_func} that the states in $\mathbb{S}$ with $r>1$ and $r\neq q$ are transient states.
Also, since a successful retransmission must be followed by a new transmission, the states in $\mathbb{S}$ with $r = q$ and $r>1$ are transient states, and the state $(2,2)$ has the action of new transmission. 
Furthermore, due to the switching structure of the optimal policy in Theorem \ref{theorem:switching}, we set a policy-switching threshold $\theta \in \mathbb{N}$ for the states with $r=1$, where the states $q>\theta, r=1$ choose the action of retransmission, while the states with $q\leq\theta, r=1$ choose the action of new transmission.
Then, it is easy to see that the states with $r=1$ and $q>\theta+1$ are transient states.
Finally, the countably infinite state space $\mathbb{S}$ is reduced to a finite state space $\mathbb{S}' = \{(2,2), (1,q), \forall q \in \{1,...,\theta+1\}$ as illustrated in Fig. \ref{fig:special_static}. 
Only the state $(1,\theta+1)$ has the action $a=1$, and the other states have the action $a=0$.

Therefore, $\theta$ is the key design parameter to be optimally designed. The policy optimization problem in the state space $\mathbb{S}$ is transformed to the one-dimensional problem. 
By calculating the stationary distribution of the states in $\mathbb{S}'$ with a given $\theta$, 
the average cost $\rho$ can be obtained, and we have the following result.
\begin{proposition} \label{proposition_special_case}
	\normalfont
In the high SNR scenario, the minimum long-term average MSE of the static channel is given as 
\begin{equation} \label{average_cost_rho}
\zeta^\star = \begin{cases}
\frac{(1-\Lambda'_0)\text{Tr}f(\bar{\mathbf{P}}_0) +(2\Lambda'_0-(\Lambda'_0)^2)\text{Tr}f^2(\bar{\mathbf{P}}_0) +(\Lambda'_0)^2 \text{Tr}f^3(\bar{\mathbf{P}}_0)}{1+\Lambda'_0} & \text{if } \theta^\star = 1\\
\frac{(1-\Lambda'_0)\bigg(\sum_{i=1}^{\theta^\star+1}\text{Tr}\left(f^{i}(\bar{\mathbf{P}}_0)\right) (\Lambda'_0)^{i-1}-\text{Tr}\left(f(\bar{\mathbf{P}}_0)\right) (\Lambda'_0)^{\theta^\star-1}\bigg)}{1-(\Lambda'_0)^{\theta^\star-1}+(\Lambda'_0)^{\theta^\star}+(\Lambda'_0)^{\theta^\star+1}} & \text{if } \theta^\star > 1\\
\end{cases}
\end{equation}
where $\theta^\star$ is the optimal policy-switching threshold.
\end{proposition}
In Proposition~\ref{proposition_special_case}, the optimal $\theta$ can be numerically obtained by linear search methods, yielding the minimum estimation MSE.

\subsection{Suboptimal Policy}
The optimal policy of the MDP problem in Sec.~\ref{sec:MDP} does not have a closed-form expression for low-complexity computation.
Besides, since the MDP problem has infinitely many states,
it has to be approximated by a truncated MDP problem with finite states for numerical evaluation and solved offline.
%Its calculation time grows exponentially as the maximum allowable number of retransmission increases. 
Therefore, we propose a easy-to-compute suboptimal policy, which is the myopic policy that makes decision simply to maximize the expected instantaneous cost.
%     Because the error covariance at receiver can be always represented as $f^q_{k}(\bar{P}_0)$ for some nonnegative integer $q$, we can map the covariance $P_k$ to $q_k$ at every time~$k$. 

Based on \eqref{static:transision_func} and \eqref{one-stage cost}, the expected next step cost $c'((r,q),a)$ given the current state $(r,q)$ and action $a$ can be derived~as
\begin{equation}
\begin{aligned}
&c'((r,q),a) \\
&=\begin{cases}
g(1)\text{Tr}\left(f^{q+1}(\bar{\mathbf{P}}_0)\right) + (1-g(1))\text{Tr}\left(f(\bar{\mathbf{P}}_0)\right), &\mbox{ if } a = 0,\\
g(r+1)\text{Tr}\left(f^{q+1}(\bar{\mathbf{P}}_0)\right) + (1-g(r+1))\text{Tr}\left(f^{r+1}(\bar{\mathbf{P}}_0)\right), &\mbox{ if } a = 1.
\end{cases}
\end{aligned}
\end{equation}
Then, we have
\begin{equation}
\begin{aligned}
&c'((r,q),1)-c'((r,q),0) \\
&= (g(r+1)-g(1))\text{Tr}\left(f^{q+1}(\bar{\mathbf{P}}_0)\right)
+(1-g(r+1))\text{Tr}\left(f^{r+1}(\bar{\mathbf{P}}_0)\right)-(1-g(1))\text{Tr}\left(f(\bar{\mathbf{P}}_0)\right).
\end{aligned}
\end{equation}
Using the inequality~\eqref{inequality_2}, $c'((r,q),1)-c'((r,q),0) \geq 0$ if and only if $(r,q)$ satisfies
\begin{equation} \label{condition}
\begin{aligned}
\text{Tr}\left(f^{q+1}(\bar{\mathbf{P}}_0)\right) 
\leq \frac{(1-g(r+1))\text{Tr}\left(f^{r+1}(\bar{\mathbf{P}}_0)\right) -(1-g(1))\text{Tr}\left(f(\bar{\mathbf{P}}_0)\right)}{g(1)-g(r+1)}.
\end{aligned}
\end{equation}
Thus, we have the following result.      
\begin{proposition} \label{prop:sub}
	\normalfont      	
	A suboptimal policy of problem~\eqref{problem} is
	\begin{equation} \label{sub_policy}
	a = \begin{cases}
	0 &\mbox{if the condition \eqref{condition} is satisfied,}\\
	1 &\mbox{otherwise.}\\
	\end{cases}
	\end{equation}
\end{proposition}

It can be proved that the suboptimal policy in Proposition~\ref{prop:sub} is also a switching-type policy. 
%Moreover, based on \eqref{sub_policy} and the monotonicity of  $\text{Tr}\left(f^{n}(\bar{P}_0)\right)$ w.r.t. $n$ discussed in Sec.~\ref{sec:estimation}, it can be verified that the action should always be zero for the states $(r,q) \in \mathbb{S}$ with $r=q$, i.e., a new transmission is required.    
\emph{Due to the simplicity of the suboptimal policy, which, unlike the optimal policy, does not need any iteration for policy calculation, it can be applied as an on-line decision algorithm. } In Sec.~\ref{sec:num}, we will show that the performance of the suboptimal policy is close to the optimal one for practical system parameters. The detailed computation-complexity analysis will be given in Sec.~\ref{sec:num} as well.

\section{Optimal Policy: Markov Channel} \label{sec:Markov}
In this section, we investigate the sensor's optimal transmission control policy in the Markov channel.

\subsection{MDP Formulation}
We also formulate the problem as a MDP.

1) The state space is defined as 
$\mathbb{S} \triangleq \{(\mathbf{\Omega}, q, \Xi) : \mathbf{\Omega} \in \mathbb{N}^B_0\backslash\mathbf{0}, q \in \mathbb{N}, \Xi \in \{1,2,\cdots,B\}\}$.
%We denote that $\left\| \mathbf{\Omega} \right\|_1 \triangleq \sum_{i=1}^{B}\delta^i$, then the total number of transmission attempts can be represented as  $r = \left\vert \mathbf{\Omega} \right\vert+1 \leq q$.

2) The action space is defined as $\mathbb{A} \triangleq \{0,1\}$. 

3) 
%As the state transition is time-homogeneous, we drop the time index $k$ here. 
Let $s=(\mathbf{\Omega},q,\Xi)$ and $s'=(\mathbf{\Omega}',q',\Xi')$ denote the current and next state, respectively. The transition probability can be written as
\begin{equation} \label{transition_function_markov}
P(s'|s,a)=
\begin{cases}
p_{\Xi,\Xi'}(1-\tilde{g}(\mathbf{0},\Xi)) & \mbox{if }a =  0, s'=(\mathbf{1}_\Xi,1,\Xi'), \\
p_{\Xi,\Xi'}\tilde{g}(\mathbf{0},\Xi) & \mbox{if } a =  0, s'=(\mathbf{1}_\Xi,q+1,\Xi'), \\
p_{\Xi,\Xi'}\left(1-\tilde{g}(\mathbf{\Omega},\Xi)\right) & \mbox{if } a =  1, s'=(\mathbf{\Omega} +\mathbf{1}_\Xi,\left\| \mathbf{\Omega} \right\|_1+1,\Xi'), \\
p_{\Xi,\Xi'}\tilde{g}(\mathbf{\Omega},\Xi) & \mbox{if } a =  1, s'=(\mathbf{\Omega} +\mathbf{1}_\Xi,q+1,\Xi'), \\
0 & \mbox{otherwise}.
\end{cases}
\end{equation}
where  $\left\| \mathbf{\Omega} \right\|_1 \triangleq \sum_{i=1}^{B}r_i$.

4) The one-stage cost is given in \eqref{one-stage cost}.

\subsection{Optimal Policy: Condition of Existence}
Inspired by the static channel scenario, we derive the following condition under which the long term average MSE can be bounded. 
\begin{theorem} \label{theorem:stability_markov}
	\normalfont
	For a Markov channel, there exists a stationary and deterministic optimal policy $\pi^*$ of problem \eqref{problem}, if the following condition holds:
	\begin{equation} \label{stability_condition_markov}
	\rho\left(\mathbf{\Pi} \mathbf{\Lambda}\right) \rho^2(\mathbf{A}) <1,
	\end{equation}
	where $\mathbf{\Pi}$ is defined in \eqref{P_matrx}, and
	\begin{equation}
	\mathbf{\Lambda} \triangleq \text{diag}\left\lbrace \Lambda_1,..., \Lambda_B \right\rbrace,
	\end{equation}
	and $\Lambda_i$ is the largest packet error probability of a retransmission when the channel power gain is $u_i$ defined in \eqref{max_error_markov}.
\end{theorem}
\begin{proof}
	See Appendix C.
\end{proof}

\begin{remark}
It is interesting to see that when retransmissions have very high reliability, i.e., ${\Lambda_i} \rightarrow 0\ \forall i=1,\cdots,B$, the eigenvalues of the matrix $\mathbf{\Pi} \mathbf{\Lambda}$ approaches to zero, and thus the left-hand side of \eqref{stability_condition_markov} is much less than one and the remote estimation system can be stabilized. 
\end{remark}
The stability regions of a two-state Markov channel in terms of $\Lambda_1$ and $\Lambda_2$ with different $\rho^2(\mathbf{A})$ are illustrated in Fig.~\ref{fig:region}, where $\mathbf{\Pi}=\begin{bmatrix}
0.8 & 0.5\\
0.2  & 0.5
\end{bmatrix}$.
We see that a larger $\rho^2(\mathbf{A})$ results in a smaller stability region.

\begin{figure*}[t]
	%	\renewcommand{\captionlabeldelim}{ }	
	%	\renewcommand{\captionfont}{\small} \renewcommand{\captionlabelfont}{\small}
	%		\hspace{-15pt}
	\minipage{0.45\textwidth}
	\centering
	%	\vspace*{-0.9cm}	
\includegraphics[scale=0.6]{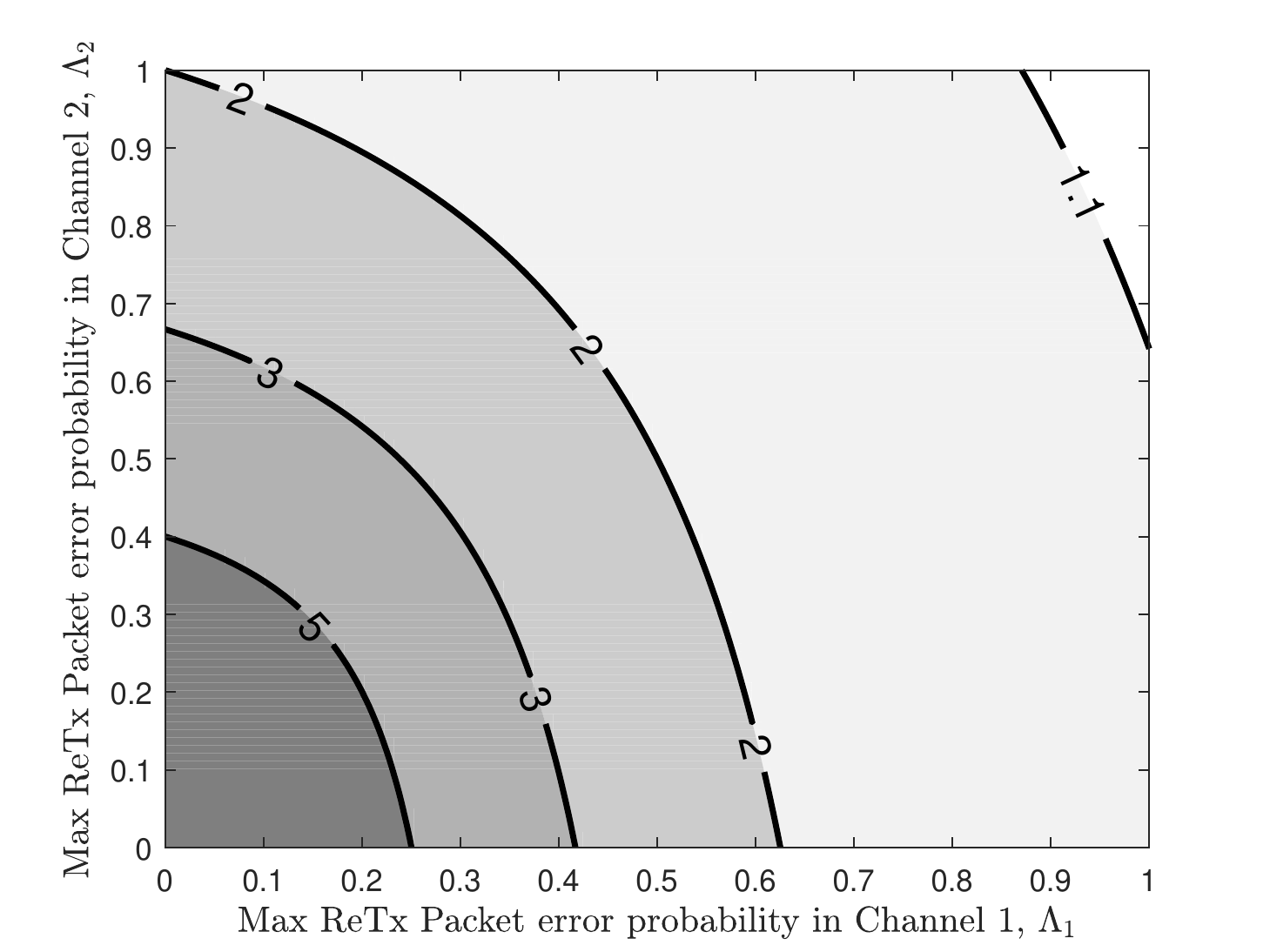}
\caption{The stability regions in terms of $\Lambda_1$ and $\Lambda_2$ with $\rho^2(\mathbf{A})=1.1,\ 2,\ 3$ and $5$, respectively.}
\label{fig:region}
	%pmfs for $\Tstd$ and $\Tuc$ with
	\endminipage
	\hspace{1cm}
	\minipage{0.45\textwidth}	
	\centering
%	\vspace*{-0.0cm}
	\includegraphics[scale=0.9]{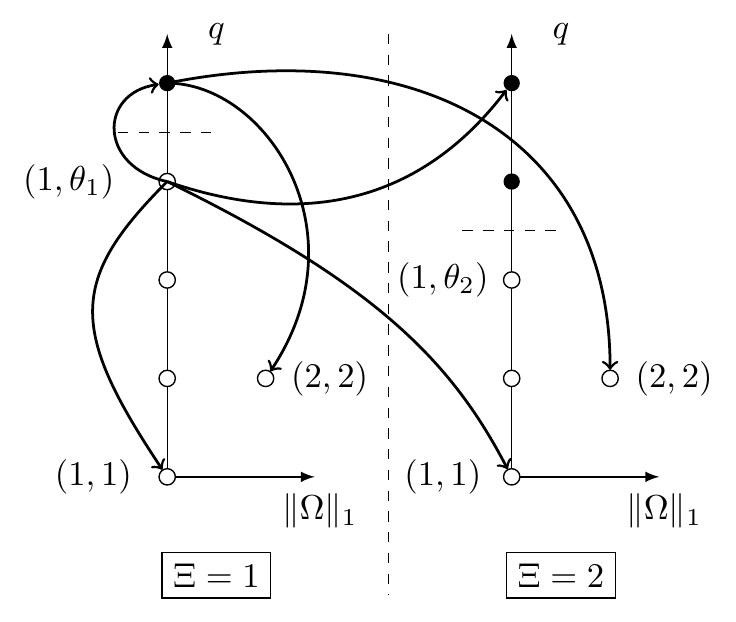}
\vspace{-0.3cm}
\caption{Illustration of the state space and state transitions  of the optimal transmission control policy in the high SNR scenario, where the channel is a $2$-state Markov channel, $\theta_1 = 4$ and $\theta_2=3$.}
\label{fig:special_markov}
	\endminipage
	%	\hrulefill
	\vspace*{-0.7cm}
\end{figure*}

%\begin{figure}[t]
%	\centering\includegraphics[scale=0.7]{region.pdf}
%	\caption{The stability regions in terms of $\Lambda_1$ and $\Lambda_2$ with $\rho^2(\mathbf{A})=1.1,\ 2,\ 3$ and $5$, respectively.}
%	\label{fig:region}
%	\vspace{-0.5cm}
%\end{figure}
%
%
%\begin{figure}[t]
%	\centering
%	\includegraphics[scale=1]{special_markov.pdf}
%	\vspace{-0.3cm}
%	\caption{Illustration of the state space and state transitions  of the optimal transmission control policy in the high SNR scenario, where the channel is a $2$-state Markov channel, and $\theta_1 = 4$ and $\theta_2=3$.}
%	\label{fig:special_markov}
%	\vspace{-0.5cm}
%\end{figure}

%\begin{proposition}
%	\normalfont
%For a Markov channel, assuming the non-retransmission policy, the long-term estimation MSE is bounded if the following condition holds:
%\begin{equation} \label{stability_condition_markov_non}
%\rho\left(\mathbf{P} \mathbf{\Lambda'}\right) \rho^2(\mathbf{A}) <1,
%\end{equation}
%where 
%\begin{equation}
%\mathbf{\Lambda'} \triangleq \text{diag}\left\lbrace \tilde{g}(\mathbf{0},1),\tilde{g}(\mathbf{0},2),..., \tilde{g}(\mathbf{0},B) \right\rbrace,
%\end{equation}
%and $\tilde{g}(\mathbf{0},i)$ is packet error probability of a new transmission when the channel power gain is $u_i$, which is defined in \eqref{Markov_channel_loss}.
%\end{proposition}

\subsection{Optimal Policy: The Structure}
The optimal policy in the Markov channel also has a switching structure in the state space.
\begin{theorem} \label{theorem:structural_markov}
	\normalfont
	(i) if $\pi^{*}(\mathbf{\Omega},q,\Xi)=0$, then  $\pi^{*}(\mathbf{\Omega} + z \mathbf{1}_i,q,\Xi)=0, \forall 1 \leq i \leq B$;
	
	(ii) if $\pi^{*}(\mathbf{\Omega},q, \Xi)=1$, then  $\pi^{*}(\mathbf{\Omega},q+z,\Xi)=1$, where $z$ is any positive integer.
\end{theorem}
\begin{proof}
		The proof is similar to that of Theorem~\ref{theorem:switching} and is omitted due to the space limitation.
\end{proof}

\subsection{Optimal Policy: A Special Case}
For the high SNR scenario, we assume that a retransmission is always successful. 
%Similar to the static channel scenario, the optimal policy has a finite state space and a switching structure. 
Thus, the state transition probability \eqref{transition_function_markov} does not depends on all the individual element of the historical channel-state vector $\mathbf{\Omega}$, and we can simply combine the states in $\mathbb{S}$ by $\| \mathbf{\Omega} \|_1$ as the state $s =(\| \mathbf{\Omega} \|_1,q,\Xi) $ to reduce the state space.

Similar to the static channel scenario, the state space of the optimal policy can be further reduced as $\mathbb{S}' = \{(2,2,\Xi), (1,q,\Xi), \forall q=1,2,...,\theta_{\max}+1,\Xi=1,2,...,B \}$, and the optimal policy for states $(1,q,\Xi), \forall q \in \{\theta_{\Xi}+1,\cdots,\theta_{\max}+1\}$ is $a=1$, where $\theta_{\max} \triangleq \max \theta_{\Xi}$, and the other states have the action $a=0$, as illustrated in Fig.~\ref{fig:special_markov}.
Different from the static channel scenario, the optimal policy for the $B$-state Markov channel has a set of parameters, i.e., $\{\theta_1,\cdots,\theta_B\}$, to be optimally designed.

We can reorder the three dimensional states as a $B\times (\theta_{\max}+2)$ state (column) vector, $\mathbf{b}$, and the states $(2,2,\Xi)$ and $(1,q,\Xi)$ are the $\left(1+(\Xi-1)(\theta_{\max}+2)\right)$th and $\left(1+q+(\Xi-1)(\theta_{\max}+2)\right)$th elements of $\mathbf{b}$, respectively.
Using the state transition probability \eqref{transition_function_markov} and the transition rule of the special case as illustrated in Fig.~\ref{fig:special_markov}, the matrix of the state transition probability can be written as
\begin{equation} \label{M_matrix}
\mathbf{M} = \left[
\begin{array}{c|c|c|c}
\mathbf{p}_1 \otimes \mathbf{M}_1 & \mathbf{p}_2 \otimes \mathbf{M}_2& \cdots & \mathbf{p}_B \otimes \mathbf{M}_B
\end{array}
\right],
\end{equation}
where the $\otimes$ is the Kronecker product operator, $\mathbf{p}_i$ is the $i$th column of $\mathbf{\Pi}$ defined in \eqref{P_matrx}, and $\mathbf{M}_i$ is the 
\begin{equation}
\mathbf{M}_i = \left[
\begin{array}{c|c}
\mathbf{E}_i&\mathbf{F}_i\\
\hline
\left[0\right]_{(\theta_{\max} -\theta_i) \times (\theta_i+2)}&\left[0\right]_{(\theta_{\max} -\theta_i) \times (\theta_{\max} -\theta_i)}
\end{array}
\right],
\end{equation}
\begin{equation}
\begin{array}{c c}
\mathbf{E}_i = 
\begin{bmatrix}
0 & 0 & 0  & \cdots &0& 1\\
1-\Lambda'_i &1-\Lambda'_i  & 1-\Lambda'_i & \cdots   & 1-\Lambda'_i & 0\\
0 & \Lambda'_i & 0 &0 &\cdots& 0\\
\Lambda'_i & 0& \Lambda'_i &\cdots & 0 & 0\\
\vdots & \vdots& \vdots &{\ddots} &\vdots& \vdots\\
0 & 0& 0 &\cdots &\Lambda'_i& 0\\
\end{bmatrix}_{(\theta_i+2)\times(\theta_i+2)}  \hspace{-1.5cm}\text{ and }
&
\mathbf{F}_i = 
\begin{bmatrix}
1 & \cdots & 1\\
0 & \cdots & 0\\
\vdots & \ddots &\vdots\\
0 & \cdots & 0\\
\end{bmatrix}_{(\theta_i+2)\times(\theta_{\max}-\theta_i)}.
\end{array}
\end{equation}
%\begin{equation}
%\mathbf{F}_i = 
%\begin{bmatrix}
%1 & \cdots & 1\\
%0 & \cdots & 0\\
%\vdots & \ddots &\vdots\\
%0 & \cdots & 0\\
%\end{bmatrix}_{(\theta_i+1)\times(\theta_{\max}-\theta_i)}.
%\end{equation}
%Note that the $(j,k)$ entry of $\mathbf{p}_i \otimes \mathbf{M}_i$ is the transition probability from state $(,i)$

Based on the stochastic matrix \eqref{M_matrix}, we can calculate the steady state distribution with a given set of policy-switching parameters. By numerically optimizing $\{\theta_1,\cdots,\theta_B\}$, we have the following result.
\begin{proposition} \label{proposition_special_case_markov}
	\normalfont
In the high SNR scenario, the minimum long-term average MSE of the Markov channel is given as 
\begin{equation}
\zeta^\star = \frac{\mathbf{c}^T \mathbf{e}}{\| \mathbf{e} \|_1},
\end{equation}
where $\mathbf{c} = \underbrace{\left[{1,\cdots,1}\right]^T}_B \otimes [\text{Tr}f^2(\bar{\mathbf{P}}_0),\text{Tr}f(\bar{\mathbf{P}}_0),\text{Tr}f^2(\bar{\mathbf{P}}_0),\cdots,\text{Tr}f^{\theta^\star_{\max}+1}(\bar{\mathbf{P}}_0)]^T$ and  $\mathbf{e}$ is a null-space vector of $\left(\mathbf{M}^{\star}-\mathbf{I}\right)$ with non-negative values, and here $\mathbf{I}$ is the $B(\theta^\star_{\max}+2)$ by $B(\theta^\star_{\max}+2)$ identity matrix.
\end{proposition}
%In Proposition~\ref{proposition_special_case_markov}, the optimal $\{\theta_1,\cdots,\theta_B\}$ can be numerically obtained, yielding the minimum estimation MSE.

\section{Numerical Results} \label{sec:num}
\subsection{Delay-Optimal Policy: A Benchmark} \label{sec:delay_optimal}
We also consider a delay-optimal policy based on the HARQ protocol in \cite{ceran2018average}, as the benchmark of the proposed optimal policy.    
We use the average AoI to measure the delay of the system.   
Therefore, similar to the MSE-optimization problem~\eqref{problem}, the delay optimization problem is formulated as
\begin{equation} \label{problem2}
\min_{\pi}\limsup_{K\to\infty}\frac{1}{K}\sum_{k=1}^{K} \mathbb{E}\left[q_k\right].
\end{equation}
This~problem can also be converted to a MDP problem with the same state space, action space and state transition function~as presented in Sec.~\ref{sec:MDP}. The one-stage cost in terms of delay~is 
\begin{equation} \label{one-stage cost2}
c((r,q),a) = q.
\end{equation}

Comparing \eqref{one-stage cost2} with \eqref{one-stage cost}, we see that the cost function of the delay-optimal policy is a linear function of $q$, while it grows exponentially fast with $q$ in the optimal policy. Thus, these two policies should be very different and their performance will be compared in the following section.   
   
\subsection{Simulation and Policy Comparison}   
In the remainder of the section, we present numerical results of the optimal policies for static and Markov channels in Sec.~\ref{sec:static channel} and Sec.~\ref{sec:Markov}, respectively, and their performance. Also, we numerically compare these optimal policies with the benchmark policy in Sec.~\ref{sec:delay_optimal}.
Unless otherwise stated, we consider CC-HARQ, and we set $\mathsf{SNR} =10$~dB, $L=100$, $R=4$,
$\mathbf{A} = \begin{bmatrix}
2.4 & 0.2 \\
0.2 & 0.8
\end{bmatrix}$, $\mathbf{C} = \begin{bmatrix}
1 & 1
\end{bmatrix}$, $\mathbf{Q}_w = I$, $\mathbf{Q}_v = 1$, 
and thus
$\rho^2(\mathbf{A}) = 1.8385^2$, 
$\bar{\mathbf{P}}_0 = \begin{bmatrix}
2.5548  &  -1.6233 \\
-1.6233 &  1.6179
\end{bmatrix}$. 

The packet error probabilities for a new transmission and a retransmission of the CC/IR-HARQ protocol are based on taking the approximation \eqref{approx_HARQ} into \eqref{error_1} and \eqref{error_2}, respectively.
We use the relative value iteration algorithm in a MDP toolbox in $\mathsf{MATLAB}$~\cite{matlab} to solve the MDP problems in Sections~\ref{sec:static channel},~\ref{sec:Markov} and~\ref{sec:delay_optimal}.
%
%
%The successful detection probability of a new transmission is $\lambda = 0.8$. It can be verified that condition~\eqref{stability_condition} holds, i.e., the optimal policy exits.

%Due to the exponential behavior of the error probability of HARQ~\cite{frenger2001performance,tripathi2003reliability},
%the packet detection error probability of a HARQ protocol is approximated as $g(r) = (1-\lambda) h^r $ for $r\geq 1$. 
%The parameter $h$ is determined by the HARQ combining scheme (e.g., the incremental redundancy scheme has a smaller $h$, i.e., a better performance, than the chase combining scheme). 
\subsubsection{Static Channel}
\ 
\\
\emph{\underline{Policy Comparison.}}
To solve the MDP problem with an infinite state space, 
the unbounded state space $\mathbb{S}$ is truncated as $\{(r, q) : 1 \leq r\leq q\leq 20\}$ to enable the evaluation. We set the channel power gain $h=2$.
Using Theorem~\ref{theorem:existence}, we can verify that the CC/IR-HARQ based optimal policy exists.
Fig.~\ref{static_policy_SNR_20} shows different policies  within the truncated state space.
In Fig.~\ref{static_policy_SNR_20}(a), we see that in line with Theorem~\ref{theorem:switching}, the optimal policy of CC-HARQ protocol is a switching-type one, where the actions of the states that are close to the states with $r=q$, are equal to zero, i.e., new transmissions are required.
Also, we see that the suboptimal policy plotted in Fig.~\ref{static_policy_SNR_20}(d) is a good approximation of the optimal one within the truncated state space. However, the delay-optimal policy plotted in Fig.~\ref{static_policy_SNR_20}(c) is very different from the previous ones, where more states have the action of new transmission.
Therefore, HARQ-based retransmissions are more important to reduce the estimation MSE than the delay.
Fig.~\ref{static_policy_SNR_20}(b) presents the optimal policy of the IR-HARQ protocol which is identical with that of CC-HARQ in Fig.~\ref{static_policy_SNR_20}(a). This is because when the channel power gain is high, e.g., $h=2$, both IR- and CC-HARQ can provide sufficiently high retransmission reliability and the transmission control policy are the same. However, we can show that the optimal policy for CC- and IR-HARQ are different when the channel power gain is low and IR-HARQ can provide much better retransmission reliability. The policy diagram is not included due to the space limitation.

\begin{figure*}[t]
	%	\renewcommand{\captionlabeldelim}{ }	
	%	\renewcommand{\captionfont}{\small} \renewcommand{\captionlabelfont}{\small}
	%		\hspace{-15pt}
	\minipage{0.45\textwidth}
	\centering
	%	\vspace*{-0.9cm}	
	\includegraphics[scale=0.55]{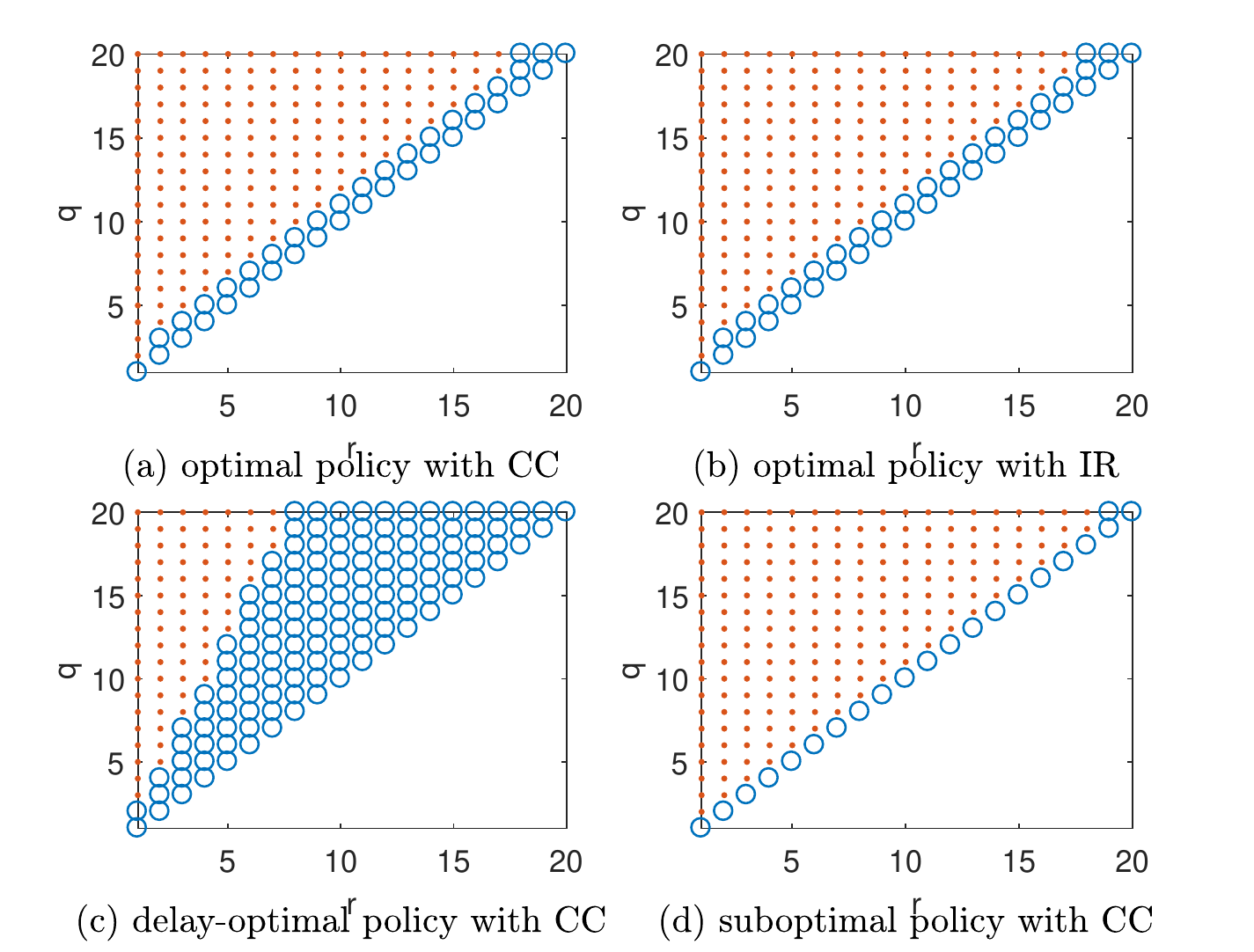}
\vspace{-1.2cm}
\caption{An illustration of different transmission control policies, where `o' and `$\cdot$' denote $a =0$ and $a=1$, respectively.}
\label{static_policy_SNR_20}
	%pmfs for $\Tstd$ and $\Tuc$ with
	\endminipage
	\hspace{0.5cm}
	\minipage{0.5\textwidth}	
	\centering
%	\vspace*{-0.8cm}
	\includegraphics[scale=0.6]{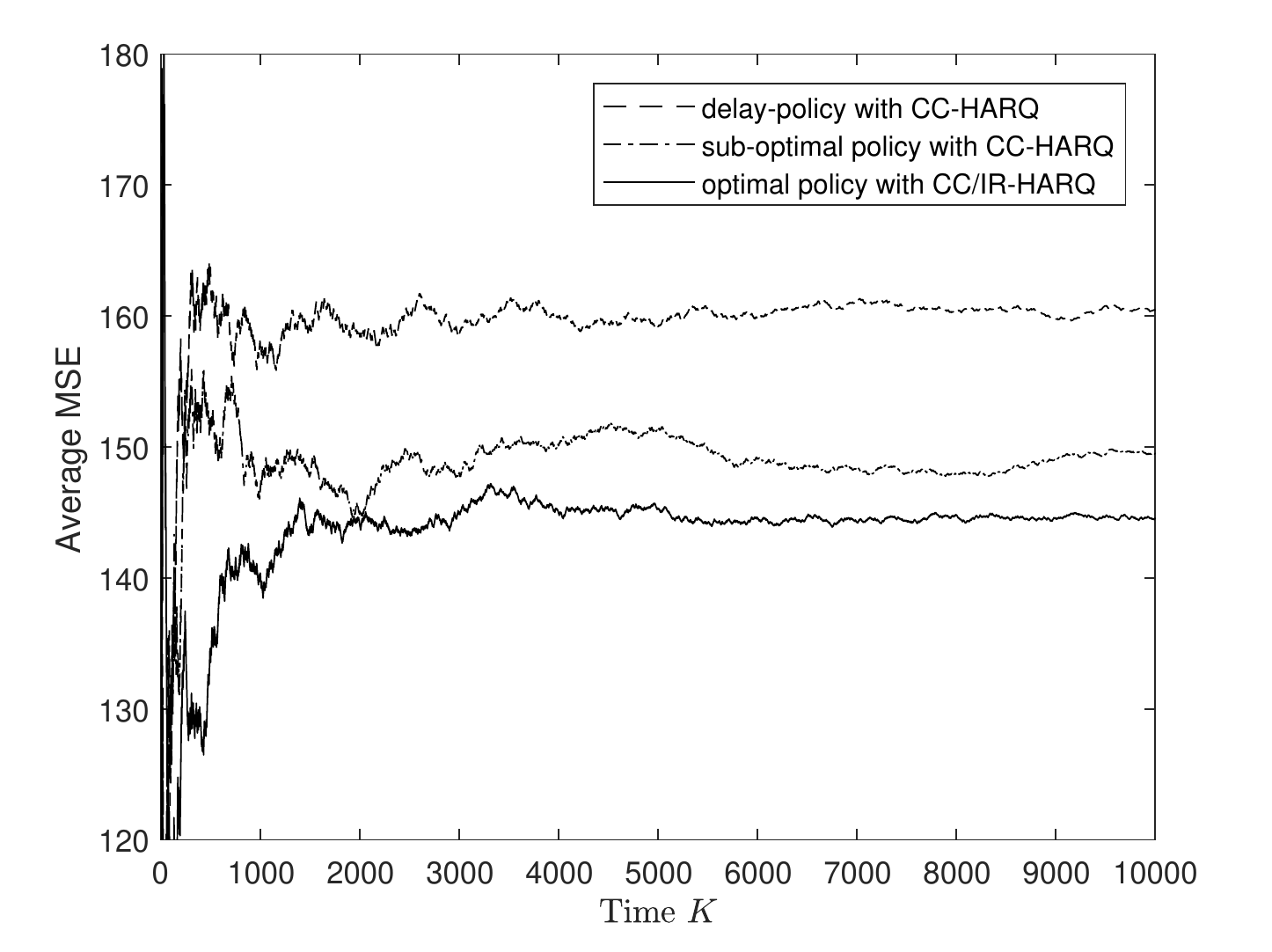}
\vspace{-1.0cm}
\caption{Average MSE with different policies of the static channel.}
\label{static_simulation_SNR_20}
	\endminipage
	%	\hrulefill
	\vspace*{-0.7cm}
\end{figure*}

%\begin{figure}[t]
%	\centering
%	\includegraphics[scale=0.65]{static_policy_SNR_20.pdf}
%	\vspace{-0.3cm}
%	\caption{An illustration of different transmission control policies, where `o' and `$\cdot$' denote $a =0$ and $a=1$, respectively.}
%	\label{static_policy_SNR_20}
%     	\vspace{-0.5cm}
%\end{figure}
%
%\begin{figure}[t]
%	\centering
%	\includegraphics[scale=0.65]{static_simulation_SNR_20.pdf}
%	\vspace{-0.3cm}
%	\caption{Average MSE with different policies, $h = 0.5$}
%	\label{static_simulation_SNR_20}
%	\vspace{-0.5cm}
%\end{figure}

\emph{\underline{Performance Comparison.}}
Based on the above numerically obtained polices and the policy with the standard ARQ, i.e., the one without retransmission (see Sec.~\ref{sec:optimal transmission control}), we further evaluate their performances in terms of the long-term average MSE using \eqref{cost_function}.
%$\frac{1}{K}\sum_{k=1}^{K} trP_k.$
%, and the delay, i.e.,
%$\frac{1}{K}\sum_{k=1}^{K} \tau_k,$
%where $K$ is the simulation time.
We run the remote estimation process with $10^4$ time slots and set the initial value of $\mathbf{P}_k$ as $\mathbf{P}_0 = f(\bar{\mathbf{P}}_0) = \begin{bmatrix}
14.2218  & -1.6966\\
-1.6966  &  1.6179
\end{bmatrix}$.
%and the initial value of $\tau_k$: $\tau_0 = 1$.
Also, we set $\text{Tr}(\mathbf{P}_0) =~15.8$ as the \emph{performance baseline}, as $\text{Tr}(\mathbf{P}_0)\leq \text{Tr}(\mathbf{P}_k)$, $\forall k$.

Fig.~\ref{static_simulation_SNR_20} plots the average MSE versus the simulation time $K$, using different transmission control policies.
Our simulation shows that the conventional non-retransmission policy has an \emph{unbounded average MSE}, which is not shown in the figure due to the ultra fast growth rate. However, we show that the average MSEs of different HARQ-based policies quickly converge to bounded steady state values. Therefore, the proposed HARQ-based policy can significantly improve the estimation quality against the conventional policy.
Also, we see that the performance of the suboptimal policy is very close to the optimal one.
Given the performance baseline, the optimal policy gives a $10\%$ MSE reduction of the delay-optimal policy, which demonstrates the superior performance of the proposed optimal policy.

%\begin{figure}[t]
%	\centering
%	\includegraphics[scale=0.5]{performance09.pdf}
%	\vspace{-0.3cm}
%	\caption{Average MSE with different policies, $h = 0.9$}
%	\label{performance_simulation}
%	     	\vspace{-0.5cm}
%\end{figure}

\begin{figure*}[t]
	%	\renewcommand{\captionlabeldelim}{ }	
	%	\renewcommand{\captionfont}{\small} \renewcommand{\captionlabelfont}{\small}
	%		\hspace{-15pt}
	\minipage{0.45\textwidth}
	\centering
	%	\vspace*{-0.9cm}	
	\includegraphics[scale=0.5]{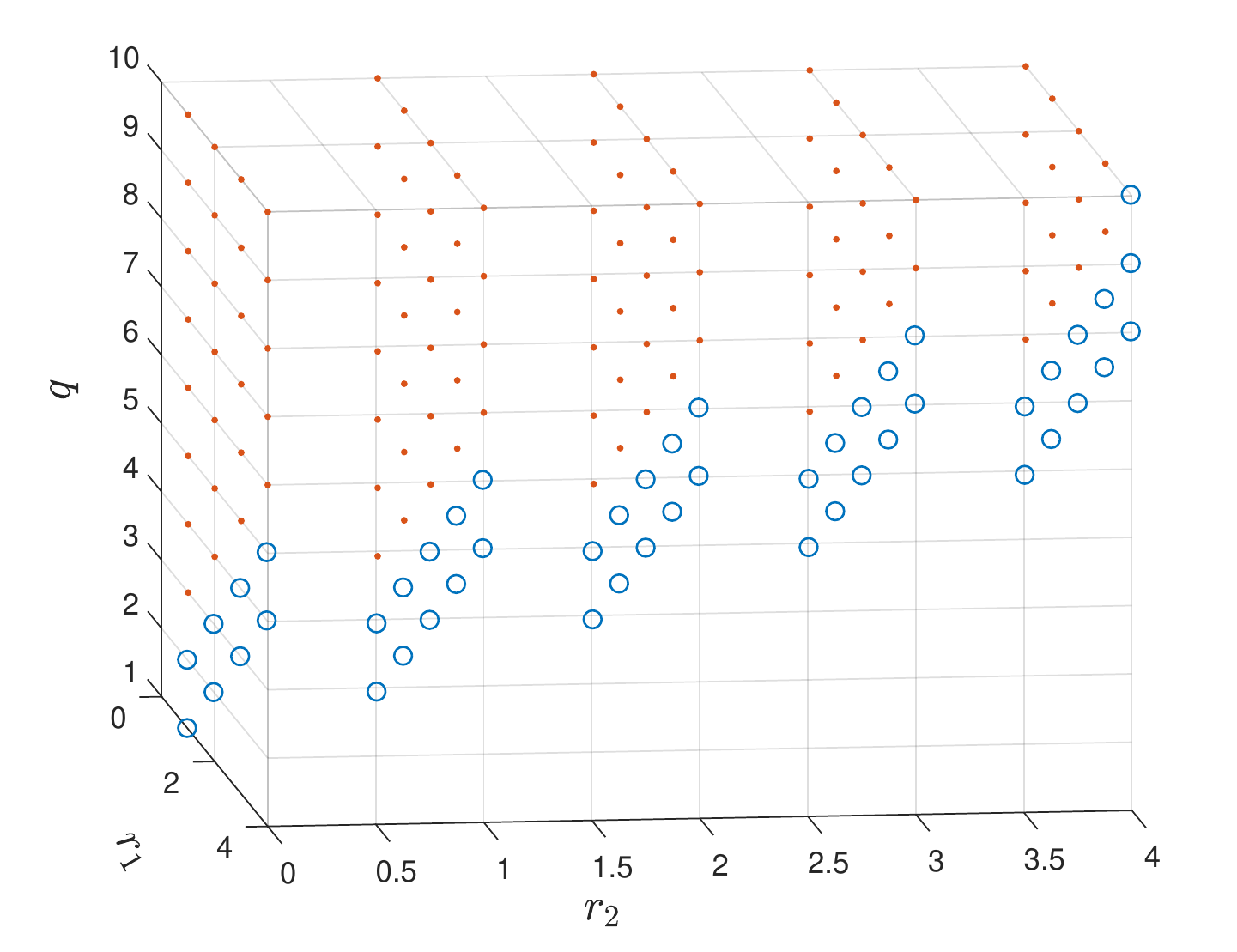}
	\vspace{-1.1cm}
	\caption{Optimal transmission control policy at channel~1, where `o' and `$\cdot$' denote $a =0$ and $a=1$, respectively.}
	\label{Markov_policy_on_channel1}
	%pmfs for $\Tstd$ and $\Tuc$ with
	\endminipage
	\hspace{1cm}
	\minipage{0.45\textwidth}	
	\centering
	%	\vspace*{-1.2cm}
	\includegraphics[scale=0.5]{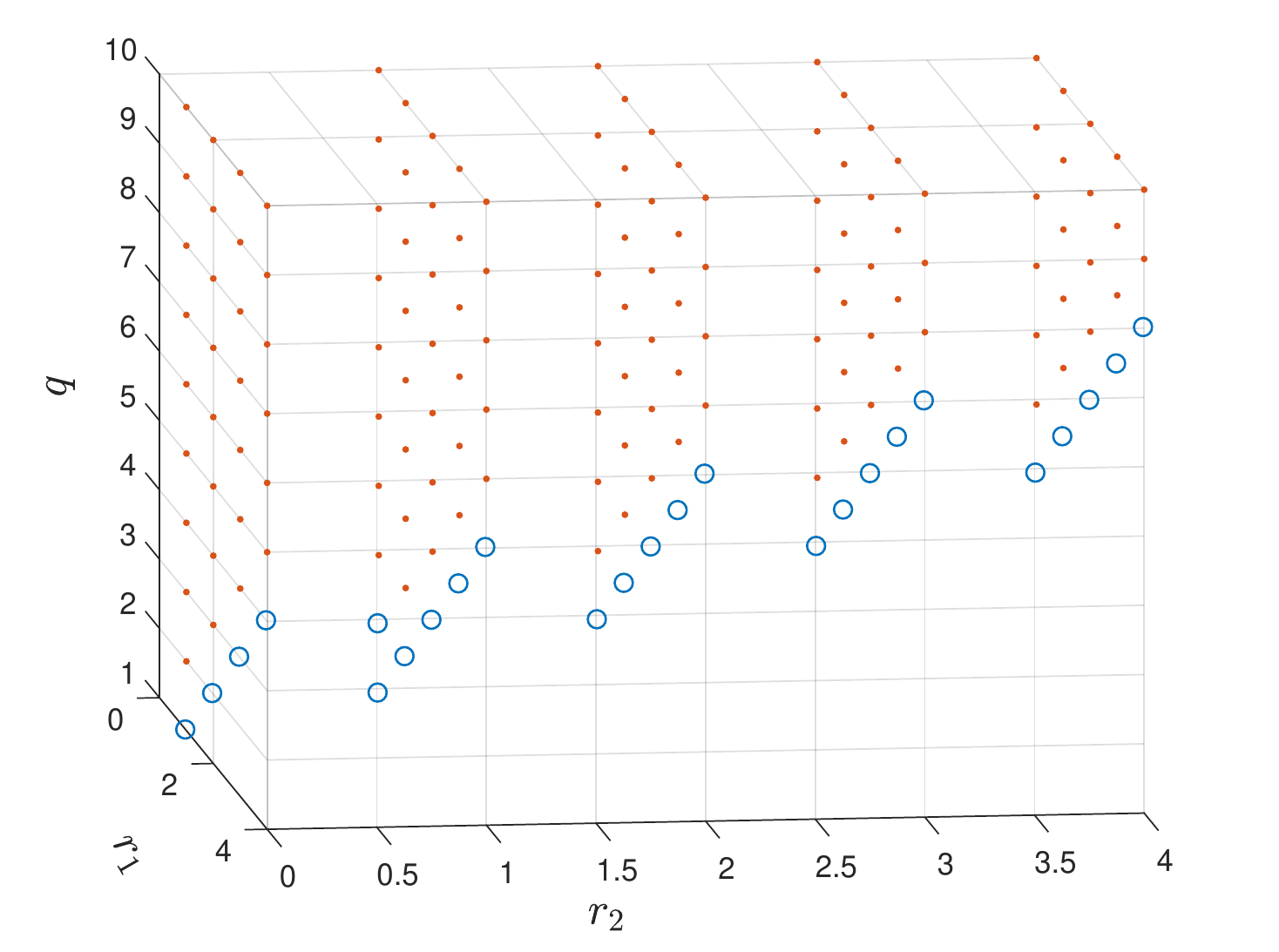}
	\vspace{-1cm}
	\caption{Optimal transmission control policy at channel~2, where `o' and `$\cdot$' denote $a =0$ and $a=1$, respectively.}
	\label{Markov_policy_on_channel2}
	\endminipage
	%	\hrulefill
	\vspace*{-0.7cm}
\end{figure*}

\subsubsection{Markov Channel}
We consider a two state Markov channel with the channel power gains $u_1 = 2$ and $u_2 =1$. The matrix of channel state transition probability $\mathbf{\Pi} = \begin{bmatrix}
0.8 & 0.2\\
0.2 &0.8
\end{bmatrix}.$
Using Theorem~\ref{theorem:stability_markov}, we can verify that the CC/IR-HARQ based optimal policy exists.
To solve the MDP problem in the Markov channel scenario with an infinite state space, 
the unbounded state space $\mathbb{S}$ is truncated as $\{(r_1,r_2, q, \Xi) : 0\leq r_1 \leq 4,0\leq r_2 \leq 4, 1\leq (r_1 + r_2 )\leq q\leq 10, \Xi =1,2\}$ to enable the evaluation, where $(r_1,r_2) \triangleq \mathbf{\Omega}$ is the state vector of the historical channel states.
Figs.~\ref{Markov_policy_on_channel1} and ~\ref{Markov_policy_on_channel2} show the optimal transmission control policy under channel 1 ($h=2$) and 2 ($h=1$), respectively.
We can see the switching structure of the optimal policy.
Also, we see that new transmission occurs more often in the good channel than in the bad channel.

We can also calculate the suboptimal policy and delay-optimal policy of the Markov channel, and the computation complexity of these policies together with the ones of the static channel are listed in Table~\ref{tab}.
We see that the numbers of convergence steps for calculating these policies are less than $100$, and the optimal policy has a larger number of convergence steps than the delay-optimal policy and hence a higher computation complexity.

%\begin{figure}[t]
%	\centering
%	\includegraphics[scale=0.65]{Markov_policy_on_channel1.pdf}
%	\vspace{-0.3cm}
%	\caption{An illustration of different transmission control policies, where `o' and `$\cdot$' denote $a =0$ and $a=1$, respectively.}
%	\label{Markov_policy_on_channel1}
%	\vspace{-0.5cm}
%\end{figure}
%
%
%\begin{figure}[t]
%	\centering
%	\includegraphics[scale=0.65]{Markov_policy_on_channel2.pdf}
%	\vspace{-0.3cm}
%	\caption{An illustration of different transmission control policies, where `o' and `$\cdot$' denote $a =0$ and $a=1$, respectively.}
%	\label{Markov_policy_on_channel2}
%	\vspace{-0.5cm}
%\end{figure}

\emph{\underline{Performance Comparison.}}
We evaluate the non-retransmission policy, CC/IR-HARQ based optimal policy, CC-HARQ based suboptimal policy and delay-optimal policy in terms of the long-term average MSE using \eqref{cost_function}.
We run the remote estimation process with $10^4$ time slots and set the initial value of $\mathbf{P}_k$ as $\mathbf{P}_0$.

Fig.~\ref{Markov_simulation} plots the average MSE versus the simulation time $K$, using different transmission control policies.
Our simulation shows that the non-retransmission policy has an \emph{unbounded average MSE}.
We show that the average MSEs of different policies quickly converge to bounded steady state values. Therefore, the proposed HARQ-based policy can also significantly improve the estimation quality against the conventional policy in the Markov channel scenario.
We see that the performance of the suboptimal policy is very close to the optimal one.
Given the performance baseline, the optimal policy reduces MSE by $33\%$ for the delay-optimal policy. 
Also, unlike the static channel scenario, we see that the IR-HARQ based optimal policy significantly reduces the average MSE by $87\%$ compared to the CC-HARQ based optimal policy. This is because IR-HARQ can provide much better retransmission reliability than CC-HARQ especially when the channel quality is bad.
\begin{table*}[t]
	\caption{Computation complexity of the policies in the static and Markov channels.}
	\vspace{-0.5cm}
	\label{tab}
	\begin{tabular}{|l|c|c|c|c|c|c|c|c|}
		\hline
		& \multicolumn{5}{c|}{Computation complexity~\cite{littman1995complexity}}                                                        & \multicolumn{3}{c|}{The number of convergence steps, $K$} \\ \hline
		&$M$& $N$	& Optimal                     & Suboptimal                  & Delay-optimal                       & Optimal      & Suboptimal    & Delay-optimal       \\ \hline
		Static channel & $2$ & $210$  & $\mathcal{O}(MN^2 K)$                    & $\mathcal{O}(1)$                      & $\mathcal{O}(MN^2 K)$                         & $58$     & 1             & $41$    \\ \hline
		Markov channel & $2$ & $328$ & $\mathcal{O}(MN^2 K)$      & $\mathcal{O}(1)$      & $\mathcal{O}(MN^2 K)$     & $90$    & 1             & $30$   \\ \hline
	\end{tabular}
	\vspace{-0.7cm}
\end{table*}
%\begin{figure*}[t]
%	%	\renewcommand{\captionlabeldelim}{ }	
%	%	\renewcommand{\captionfont}{\small} \renewcommand{\captionlabelfont}{\small}
%	%		\hspace{-15pt}
%	\minipage{0.5\textwidth}
%	\centering
%	%	\vspace*{-0.9cm}	
%	\includegraphics[scale=0.65]{Markov_simulation.pdf}
%\vspace{-0.9cm}
%\caption{Average MSE with different policies of the Markov channel}
%\label{Markov_simulation}
%	%pmfs for $\Tstd$ and $\Tuc$ with
%	\endminipage
%	\hspace{1cm}
%	\minipage{0.4\textwidth}	
%	\centering
%	%	\vspace*{-1.2cm}
%\includegraphics[scale=1]{prooftheorem1.pdf}
%\vspace{-0.2cm}
%\caption{Illustration of the state transition paths, $(i,i)\rightarrow (1,1)$ and $(i,i)\rightarrow (i-1,i-1)$.}
%\label{fig:proof1}
%	\endminipage
%	%	\hrulefill
%	\vspace*{-0.7cm}
%\end{figure*}
%
\begin{figure}[t]
	\centering
	\includegraphics[scale=0.65]{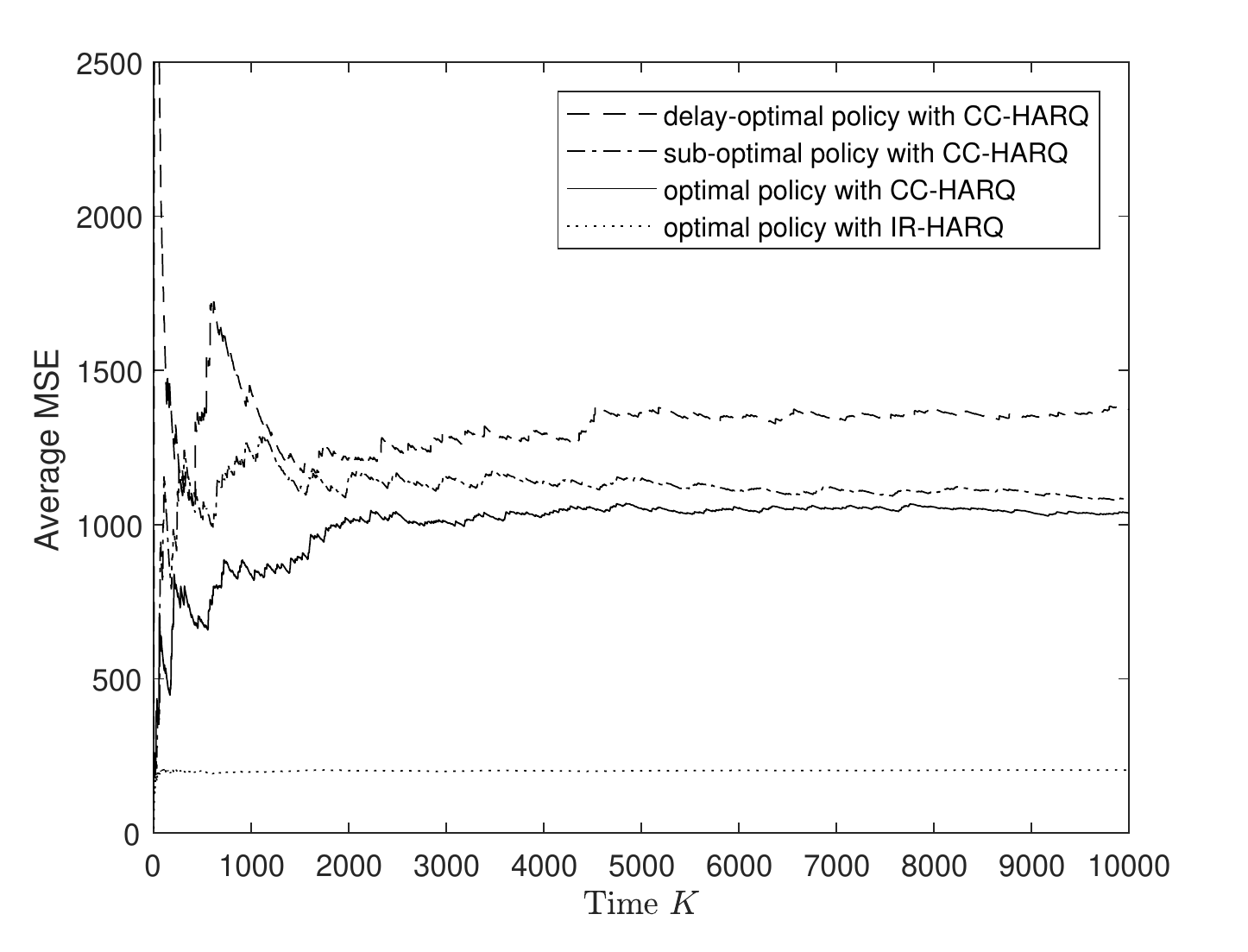}
	\vspace{-0.5cm}
	\caption{Average MSE with different policies of the Markov channel}
	\label{Markov_simulation}
	\vspace{-0.7cm}
\end{figure}

\section{Conclusions}\label{sec:con}
We have proposed and optimized a HARQ-based remote estimation protocol for real-time applications.
Our results have shown that the optimal policy can significantly reduce the estimation MSE for some practical settings.
As the recent
communication standards for real-time wireless control, such
as WirelessHART, ISA-100 and IEEE 802.15.4e, have not
adopted any HARQ techniques,
this work also suggests that HARQ can be adopted by the future real-time communication standards to enhance the performance of mission-critical remote estimation/control systems.
%For future work, we will consider the problem of optimal policy design with fading and Markov channel.

     \section*{Appendix A: Proof of Theorem~\ref{theorem:existence}}
	To prove the existence of a stationary and deterministic optimal policy given condition~\eqref{stability_condition}, we need to verify the following {conditions}~\cite[Corollary 7.5.10]{sennott2009stochastic}:
	(CAV*1) there exists a standard policy $\psi$ such that the recurrent class $R_{\psi}$ induced by $\psi$ is equal to the whole state space $\mathbb{S}$;
    (CAV*2) given $U > 0$, the set $\mathbb{S}_U = \{s|c(s,a) \leq U  \mbox{ for some } a\}$ is finite.	    
   
    Condition (CAV*2) can be easily verified based on \eqref{one-stage cost}. In what follows, we verify (CAV*1) by first constructing a policy $\psi$ and then proving that it is a standard policy.
    
    The action of the policy $\psi$ is given as 
    \begin{equation} \label{standard}
    a = \begin{cases}
    0, &\text{if}\ r=q\\
    1, &\text{otherwise}.
    \end{cases}
    \end{equation}
    It is easy to prove that any state in $\mathbb{S}$ induced by $\psi$ is a recurrent state.
    We then prove that $\psi$ is a standard policy by verifying both the expected first passage cost and time from state $(r,q) \in \mathbb{S}$ to $(0,0)$ are bounded~\cite{sennott2009stochastic}. Due to the space limitation, we only prove that any state with $r=q$ has bounded first passage cost and time. The other states can be proved similarly. 
 \begin{figure}[t]
	\centering\includegraphics[scale=1]{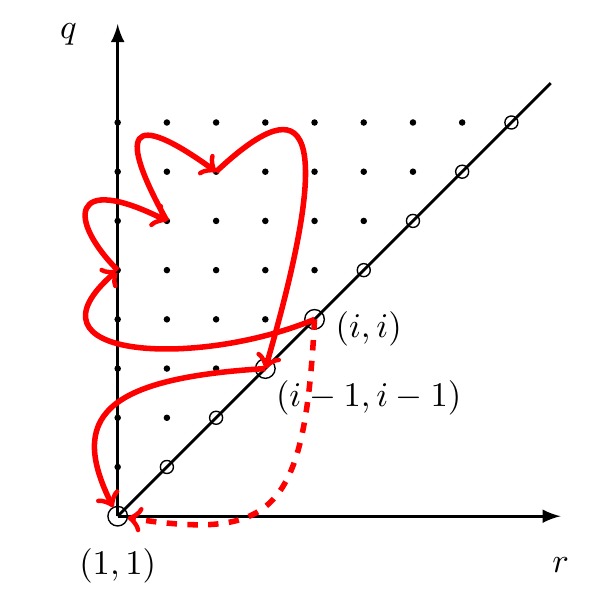}
	\vspace{-0.5cm}
	\caption{Two of the possible first passage paths from $(i,i)$ to $(1,1)$, i.e., $(i,i)\rightarrow (1,1)$ and $(i,i)
		\rightarrow (1,i+1)
		\rightarrow (2,i+2)
		\rightarrow (3,i+3)
		\rightarrow (i-1,i-1)
		\rightarrow (1,1)$, where $i=5$.}
	\label{fig:proof1}
	\vspace{-0.7cm}
\end{figure}
    For notational simplicity, the expected first passage cost of the state $(i,i)$ is denoted as $d(i)$, and the one-stage cost \eqref{one-stage cost} is rewritten as    
    \begin{equation}
    c(q) \triangleq c((r,q),a) = \text{Tr}\left(f^{q}(\bar{\mathbf{P}}_0)\right).
    \end{equation} 
    Based on \eqref{static_channel_loss}, \eqref{standard} and the law of total expectation of the first passage cost of all the possible first passage paths (as illustrated in Fig.~\ref{fig:proof1}), the expected first passage cost $d(i)$ can be obtained~as
\begin{equation} \label{equation_group}
\begin{aligned}
&d(i) = c(i) + (1-g(1))c(1) + g(1) c(i+1)\\
&+ g(1) (1-g(2)) d(2) + g(1) g(2) c(i+2)\\
&+ g(1) g(2) (1-g(3)) d(3) + g(1) g(2) g(3) c(i+3) + \cdots\\
&= \nu(i) +(1-g(1))c(1)+ D, \forall i >0,
%    &= c(i) + (1-g(1))c(0) + \sum_{j=1}^{\infty} \prod_{l=1}^{j} g(l-1)(1-g(j))d(j)\\
%    &+ \sum_{j=1}^{\infty} \prod_{l=1}^{j} g(l-1)c(i+j),\  \forall i >0.
\end{aligned}
\end{equation}
where $g(1) = \Lambda'_0$,
\begin{align}
\label{h}
&\nu(i) = c(i) +\sum_{j=1}^{\infty} \alpha_j c(i+j),\ D= \sum_{j=2}^{\infty} \beta_j d(j),
\end{align}
and $\alpha_j = \prod_{l=1}^{j} g(l)$ and $\beta_j = \prod_{l=1}^{j-1} g(l)(1-g(j))$.
Therefore, $d(i)$ is bounded if $\nu(i)< \infty$ and $D< \infty$. Using the inequality~\eqref{inequality_2}, we have 
\begin{equation} \label{inequality_alphabeta}
\begin{aligned}
&\alpha_j \leq \Lambda'_0\Lambda^{j-1}_0, &&j\geq 1, \\
&\beta_j \leq \Lambda'_0\Lambda^{j-2}_0, &&j\geq 2.
\end{aligned}
\end{equation}

From~\cite{schenato2008optimal}, we have $\sum_{j=1}^{\infty} \Lambda^j_0c(j) < \infty$ iff $\Lambda_0 \rho^2(\mathbf{A})<1$.
Thus, it is easy to prove that $\nu(i) < \infty$ if \eqref{stability_condition} holds.

From \eqref{equation_group}, $D$ can be further derived after simplifications as
\begin{equation}
D = \frac{1}{1-\sum_{i=2}^{\infty}\beta_i} \left(\sum_{i=2}^{\infty} \beta_i (1-g(1))c(1) + \sum_{i=2}^{\infty} \beta_i \nu(i) \right).
\end{equation}
As $\sum_{i=2}^{\infty}\! \beta_i < g(1)\!<\! 1$, $D$ is bounded as long as $\sum_{i=2}^{\infty} \!\beta_i \nu(i) \!<\! \infty$. Using the inequalities~\eqref{inequality_alphabeta}, after some simplifications, we have
\begin{equation}
\begin{aligned}\label{temp}
\sum_{i=2}^{\infty} \beta_i \nu(i) 
\leq 
\eta \sum_{j=2}^{\infty} \Lambda^{j-1}_0 c(j)+
\eta^2\sum_{j=3}^{\infty}  (j-2) \Lambda^{j-1}_0 c(j),
\end{aligned}
\end{equation}
where $\eta = \Lambda'_0/\Lambda_0$.
It can be proved that $\sum_{j=3}^{\infty}  (j-2) \Lambda^{j-1}_0 c(j)$ is bounded if $\sum_{j=1}^{\infty} \Lambda^{j}_0 c(j)$ is bounded. Again, using the result that 
$\sum_{j=1}^{\infty} \Lambda^{j}_0 c(j) < \infty$ iff $\Lambda_0 \rho^2(\mathbf{A})<1$ in \cite{schenato2008optimal}, $\sum_{i=1}^{\infty} \beta_i \nu(i)<\infty$ if $\Lambda_0 \rho^2(\mathbf{A})<1$, yielding the proof of the bounded expected first passage cost with condition~\eqref{stability_condition}. Similarly, we can verify that the expected first passage time
is also bounded.

    \section*{Appendix B: Proof of Theorem~\ref{theorem:switching}}
%     Assume the existence of an optimal policy based on Theorem~\ref{theorem:existence}.
%     $\pi^*$ can be computed from following average cost optimal equation (ACOE) \cite{sennott2009stochastic}:
%    \begin{equation}
%    \rho + h(s) = min_{\pi \in \Pi}\{c(s,\pi(s))+\sum_{s'}P_{ss'}h(s')\}
%    \end{equation}
%    where $\rho$ is a constant, $P_{ss'}$ is the transition probability from state $s$ to $s'$, and $h(\cdot)$ is a function on $\mathbb{S}$, which is also noted as relative value function.     
    	The switching property is equivalent to the monotonicity of the optimal policy in $r$ if $q$ is fixed and in $q$ if $r$ is fixed. The monotonicity can be proved by verifying the following conditions (see Theorem 8.11.3 in \cite{puterman2014markov}).
    
    (1) $c(s,a)$ is nondecreasing in $s$ for all $a \in \mathbb{A}$;
    
    (2) $c(s,a)$ is a superadditive function on $\mathbb{S} \times \mathbb{A}$;
    
    (3) $\tilde{P}(s'|s,a)= \sum_{i=s'}^{\infty} P\left(i|s,a\right)$ is nondecreasing in $s$ for all $s'\in \mathbb{S}$ and $a \in \mathbb{A}$;
    
    (4)  $\tilde{P}(s'|s,a)$ is a superadditive function on $\mathbb{S} \times \mathbb{A}$ for all  $s'\in \mathbb{S}$.
    
    We first prove the monotonicity in $r$ with $q$ fixed. The state $s$ is ordered by $r$, i.e., if $r^-\leq r^+$, we define $s^- \leq s^+$ with $s^-=(r^-,q)$ and $s^+=(r^+,q)$. 
    %Similarly,  the state $s$ is ordered by $q$, i.e., if $q\leq q'$, we define $s \preceq s'$ with $s=(r,q)$ and $s'=(r',q')$.
	From the definition of one-stage cost, $c(s,a)$ is increasing in $q$. Therefore, condition (1) can be easily verified.        
    For condition (2), the superadditive function is defined in (4.7.1) of \cite{puterman2014markov}. A function $f(x,y)$ is superadditive for $x^{-} \leq x^{+}$ and $y^{-} \leq y^{+}$, if
    $f(x^{+},y^{+})+f(x^{-},y^{-}) \geq f(x^{+},y^{-})+f(x^{-},y^{+})$. Then, condition (2) can be easily verified as $c(s,a)$ is independent of $a$.
    
    Given the current state $s = (r,q)$, from \eqref{static:transision_func}, the next possible states are $s_0 \triangleq (0,0)$, $s_1 \triangleq (0,q+1)$, $s_2 \triangleq (r+1,r+1)$ and $s_3 \triangleq (r+1,q+1)$.
   Let $s'\triangleq\left\lbrace(r',q'): q\in \mathbb{N} \right\rbrace$.
    If $r'\leq r$, we define $s' \preceq s$ with $s=(r,q)$. 
    Based on \eqref{static:transision_func}, $\tilde{P}(s'|s,a)$ with different actions are given as
\begin{equation}
\tilde{P}(s'|s,a\!=\!0)\!=\!\begin{cases}
1, &\mbox{if } s' \preceq s_0\\
0, &\mbox{otherwise }
\end{cases},
\end{equation}
and
\begin{equation}
\tilde{P}(s'|s,a\!=\!1)\!=\!\begin{cases}
1, &\mbox{if } s' \preceq s_2\\
0, &\mbox{otherwise}
\end{cases}.
\end{equation}
%      
%$
%    q(s'|s,a\!=\!0)\!=\!\begin{cases}
%    1, &\mbox{if} s' \leq s_0\\
%    g(0), &\mbox{if } s_0 < s' \leq s_1\\
%    0, &\mbox{if } s' > s_1
%    \end{cases}
%$
%and
%$
%    q(s'|s,a\!=\!1)\!=\!\begin{cases}
%    1, &\mbox{if} s' \leq s_2\\
%    g(r+1), &\mbox{if } s_2 < s' \leq s_3\\
%    0, &\mbox{if } s' > s_3
%    \end{cases}
%$.   
    Therefore, condition (3) can be easily verified.
    
    For condition (4),    
    let $s^{+} = (r^{+},q)$, $s^{-} = (r^{-},q)$,  $r^{+} \geq r^{-}$ and  $a^{+} \geq a^{-}$ Then, we need to verify if
    $
    \tilde{P}(s'|s^{+},a^{+})+\tilde{P}(s'|s^{-},a^{-}) \geq \tilde{P}(s'|s^{+},a^{-})+\tilde{P}(s'|s^{-},a^{+}).
    $
    Based on the definitions of $\tilde{P}(s'|s,a)$, $s'$ and $s_i$, $i=0,1,2,3$, %the fact that $s_0^{-}=s_0^{+} \leq s_1^{-}=s_1^{+} \leq s_2^{-} \leq s_3^{-} \leq s_2^{+} \leq s_3^{+}$,
    condition (4) can be verified straightforwardly.
	As all four conditions hold, the monotonicity of the optimal policy in $r$ is proved. Similarly, the monotonicity of the optimal policy in $q$ can be proved. 

\section*{Appendix C: Proof of Theorem \ref{theorem:stability_markov}}

The proof consists of three steps: 1) construction of a stationary policy in state space $\mathbb{S}$, 2) 
providing useful technical lemmas for problem transformation
and 3) deriving a sufficient condition in terms of the packet error probability such that the long-term average cost of the stationary policy is bounded, completing the proof of existence condition of Theorem~\ref{theorem:stability_markov}.

Step 1. Inspired by the proof of the static channel scenario, where the constructed stationary policy \eqref{standard} is simply the policy that a retransmission is always required until the a successful transmission occurs, we consider a similar policy in the Markov channel scenario as
\begin{equation} \label{standard2}
a = \psi(\mathbf{\Omega},q,\Xi) = \begin{cases}
0, &\text{if}\ \| \mathbf{\Omega} \|_1=q\\
1, &\text{otherwise},
\end{cases}
\end{equation}
We can prove that $\psi$ is a stationary policy in the state space $\mathbb{S}$.
In what follows, we prove that the long-term average cost induced by policy $\psi$ is bounded if \eqref{stability_condition_markov} holds, completing the proof of existence condition of Theorem~\ref{theorem:stability_markov}.

As the state $s=(\mathbf{\Omega},q,\Xi)$ has $B+2$ dimensions, it is not possible to analyze the average cost directly and thus we have to reduce the dimension of the state space.

Step 2. Some useful lemmas.
\begin{lemma}
	\normalfont
	Given the policy $\psi$ and the packet loss function 
	\begin{equation} \label{Markov_channel_loss'}
	\mathbb{P}\left[\gamma_k = 0  \right] 
	=
	\left\lbrace \begin{aligned}
	& \tilde{g}(\mathbf{0},\Xi_{k}), &&a_k = 0\\
	& \tilde{g}'(\mathbf{\Omega}_k,\Xi_{k}), &&a_k = 1
	\end{aligned}
	\right.
	\end{equation}
	where $\tilde{g}(\mathbf{\Omega}_k,\Xi_{k})< \tilde{g}'(\mathbf{\Omega}_k,\Xi_{k})< 1$, $\forall k$,
	the average cost of the MDP with the packet loss function \eqref{Markov_channel_loss} is bounded, if that of \eqref{Markov_channel_loss'} is.
\end{lemma}
The proof is straightforward due to the fact that a larger packet error probability results in a larger average cost, and thus is omitted here.

In the following, we derive a sufficient condition that can stabilize the average cost of the MDP with the packet loss function
\begin{equation} \label{'}
\mathbb{P}\left[\gamma_k = 0  \right] 
=
\left\lbrace \begin{aligned}
& \Lambda'_i = \tilde{g}(\mathbf{0},\Xi_{k}), &&a_k = 0, \Xi_k=i\\
& \Lambda_i \geq \tilde{g}(\mathbf{\Omega}_k,\Xi_{k}), &&a_k = 1, \Xi_k=i.
\end{aligned}
\right.
\end{equation}
Since the packet error function \eqref{'} does not depends on the individual elements of the state vector $\mathbf{\Omega}_k$, the state space $\mathbb{S}$ is reduced to a three-dimensional space $\mathbb{S}' \triangleq \{(r, q, \Xi) : r \leq q,\ (r, q) \in \mathbb{N} \times \mathbb{N}, \Xi\in \{1,2,...,B\}\}$, where $r\triangleq \|\mathbf{\Omega}\|_1$.

Again using \cite[Corollary 7.5.10]{sennott2009stochastic}, we only need to proof that both (CAV*1) and (CAV*2) hold in the new state space $\mathbb{S}'$ under the condition \eqref{stability_condition_markov}. 
Similar to the proof of Theorem~\ref{theorem:existence}, it can be easily proved that the recurrent class $R_{\psi}$ induced by $\psi$ is equal to the whole state space $\mathbb{S}'$, and (CAV*2) holds. Thus, in what follows, we only need to prove that the policy $\psi$ is a standard policy, i.e., both the expected first passage cost and time from any state in $\mathbb{S}'$ to the state $(1,1,1)$ is bounded if \eqref{stability_condition_markov} is satisfied.

\begin{lemma} \label{lem:combine_states}
	\normalfont
	If the first passage cost from any state $s \in \mathbb{S}'$ to 
	the set of states 
	\begin{equation}
	\mathbf{s} \triangleq \{(1,1,1),(1,1,2),...,(1,1,B)\}
	\end{equation} 
	is bounded, the first passage cost from any state $s \in \mathbb{S}'$ to
	state $(1,1,1)$ is bounded.
\end{lemma}
\begin{proof}
Let $A_{s,\mathbf{s}}$ be the expected first passage cost from $s$ to the set $\mathbf{s}$. We have
\begin{equation}
A_{s,\mathbf{s}} = \sum_{i=1}^{B} \varrho_{s,i} A_{s,i},
\end{equation}
where 	
$A_{s,i}$ is the expected first passage cost from $s$ to the set $\mathbf{s}$ with the condition that the first visited state in $\mathbf{s}$ is $(1,1,i)$, and $\varrho_{s,i}$ is the probability that the first visited state from $s$ to the set $\mathbf{s}$ is $(1,1,i)$.

As the state transition probabilities $p_{i,j}$ in \eqref{P_matrx} are larger than zero, it is clear that $\varrho_{s,i} \in (0,1)$. Therefore, if $A_{s,\mathbf{s}}$ is bounded $\forall s\in \mathbf{s}$, $A_{s,i}$ is bounded $\forall s\in \mathbf{s}, i\in\{1,...,B\}$.

The average first passage cost from $s$ to the state $(1,1,1)$ can be written as
\begin{equation} \label{U_inequality}
\begin{aligned}
&U_{s,1} 
= \varrho_{s,1} A_{s,1} 
+ \sum_{i=2}^{B} \varrho_{s,i}\varrho_{i,1} \left(A_{s,i}+A_{i,1}\right)
+ \sum_{j=2}^{B}\sum_{i=2}^{B} \varrho_{s,j}\varrho_{j,i}\varrho_{i,1} \left(A_{s,j}+A_{j,i}+A_{i,1}\right) \\
&+ \sum_{j=2}^{B}\sum_{i=2}^{B}\sum_{k=2}^{B} \varrho_{s,j}\varrho_{j,k}\varrho_{k,i}\varrho_{i,1} \left(A_{s,j}+A_{j,k}+A_{k,i}+A_{i,1}\right)+ \cdots\\
&< \varrho_{s,1} A_{\max} +(1-\varrho_{s,1}) 2 A_{\max}  +(1-\varrho_{s,1})(1-\varrho_{\min})  3 A_{\max}  +(1-\varrho_{s,1})(1-\varrho_{\min})^2  4 A_{\max} \cdots
\end{aligned}
\end{equation}
where $A_{i,j}$ is the expected first passage cost from $(1,1,i)$ to $(1,1,j)$, and the $\varrho_{j,k}$ is the probability that the first visited state from $(1,1,j)$ to the set $\mathbf{s}$ is $(1,1,k)$.
\begin{equation}
A_{\max} = \max_{i,j = 1,2,...,B }\{A_{s,i},A_{i,j}\},
\text{ and }
\varrho_{\min} = \min_{i=1,2,...,B} \{\varrho_{i,1}\}.
\end{equation}
Since $(1-\varrho_{\min})<1$ and $A_{\max}<\infty$, it is straightforward that the right-hand side of the inequality \eqref{U_inequality} is bounded, completing the proof.
\end{proof}

Using Lemma~\ref{lem:combine_states}, in the following, we only need to prove that if \eqref{stability_condition_markov} holds, the average first passage cost from any state to the state set $\mathbf{s}$ is bounded, which is also a sufficient condition that guarantees the average first passage cost from any state to a specific state, say $(1,1,1)$, is bounded.

Step 3.
We define a length-$B$ row vector $\vec{d}(i), \forall i\in \mathbb{N}$, where the $k$th element of $\vec{d}(i)$ is the average first passage cost from state $(i,i,k)$ to the state set $\mathbf{s}$.
Analogously to the static channel scenario~\eqref{equation_group}, we have the following equation
\begin{equation} \label{equation_group2}
\begin{aligned}
&\vec{d}(i) = \vec{c}(i) + \vec{c}(1)(\mathbf{\Pi}(\mathbf{I}-\mathbf{\Lambda}')) + \vec{c}(i+1)\mathbf{\Pi}\mathbf{\Lambda}' \\
&+ \vec{d}(2) (\mathbf{\Pi}(\mathbf{I}-\mathbf{\Lambda})) \mathbf{\Pi}(\mathbf{\Lambda}')   + \vec{c}(i+2)(\mathbf{\Pi}(\mathbf{\Lambda})) \mathbf{\Pi}(\mathbf{\Lambda}') \\
&+ \vec{d}(3)(\mathbf{\Pi}(\mathbf{I}-\mathbf{\Lambda}))(\mathbf{\Pi}(\mathbf{\Lambda})) \mathbf{\Pi}(\mathbf{\Lambda}')  + \vec{c}(i+3)(\mathbf{\Pi}(\mathbf{\Lambda}))(\mathbf{\Pi}(\mathbf{\Lambda})) \mathbf{\Pi}(\mathbf{\Lambda}')  + \cdots\\
&= \vec{\nu}(i) +\vec{c}(1) \mathbf{\Pi}(\mathbf{I}-\mathbf{\Lambda}')+ \vec{D}, \forall i >0,
\end{aligned}
\end{equation}
where $\mathbf{I}$ is the $B$-by-$B$ identity matrix, 
$
\mathbf{\Lambda}' \triangleq \text{diag}\left\lbrace \Lambda'_1,..., \Lambda'_B \right\rbrace$, 
$\vec{c}(i) \triangleq {c}(i) \vec{1}$, 
$\vec{1}\triangleq  \underbrace{\left[1,1,...,1\right]}_{B}$,
and 
\begin{align}
\label{h2}
&\vec{\nu}(i) = \vec{c}(i) +\sum_{j=1}^{\infty}  \vec{c}(i+j) \check{\alpha}_j,\\
& \vec{D}= \sum_{j=2}^{\infty}  \vec{d}(j) \check{\beta}_j,\\  \label{alpha}
& \check{\alpha}_j =  \left(\mathbf{\Pi}\mathbf{\Lambda}\right)^{j-1}\mathbf{\Pi}\mathbf{\Lambda'}, j \geq 1\\ \label{beta}
&\check{\beta}_j=\mathbf{\Pi}\left(\mathbf{I}-\mathbf{\Lambda}\right)\left(\mathbf{\Pi}\mathbf{\Lambda}\right)^{j-2}\mathbf{\Pi}\mathbf{\Lambda'}, j\geq 2.
\end{align}
Therefore, $\vec{d}(i)$ is bounded if $\nu(i)< \infty$ and $D< \infty$. 

From the definitions in \eqref{alpha} and \eqref{beta}, we have the inequalities
\begin{equation} \label{inequality_alphabeta2}
\begin{aligned}
&\check{\alpha}_j \preceq \kappa_1 \left[\rho^{j-1}(\mathbf{\Pi}\mathbf{\Lambda})\right]_{B \times B}, &&j\geq 1\\
&\check{\beta}_j \preceq \kappa_2 \left[\rho^{j-2}(\mathbf{\Pi}\mathbf{\Lambda})\right]_{B \times B}, &&j\geq 2\\
&\check{\alpha}_i \check{\beta}_j \preceq \kappa_3 \left[\rho^{i+j-3}(\mathbf{\Pi}\mathbf{\Lambda})\right]_{B \times B}, &&i\geq 1, j\geq 2
\end{aligned}
\end{equation}
where $\kappa_1$, $\kappa_2$ and $\kappa_3$ are positive constant.
From~\cite{shi2012optimal}, we have the inequality of the cost function
\begin{equation}\label{inequal_c}
c(j)\leq \kappa \rho^{2j}(\mathbf{A}),
\end{equation}
where $\kappa$ is a positive constant.
Thus, using the inequalities \eqref{inequality_alphabeta2} and \eqref{inequal_c}, it is easy to prove that $\vec{\nu}(i) < \infty$ if \eqref{stability_condition_markov} holds.

From \eqref{equation_group2}, we further have
\begin{equation} \label{D}
\vec{D} \left(\mathbf{I}-\sum_{i=2}^{\infty}\check{\beta}_i\right)  = \vec{c}(1) \mathbf{\Pi}(\mathbf{I}-\mathbf{\Lambda}') \sum_{i=2}^{\infty} \check{\beta}_i + \sum_{i=2}^{\infty} \vec{\nu}(i)  \check{\beta}_i.
\end{equation}
Using the result of sum of a geometric series of matrices, it can be obtained as
\begin{equation}
\sum_{i=2}^{\infty}\check{\beta}_i = 
\mathbf{\Pi}(\mathbf{I}-\mathbf{\Lambda})(\mathbf{I}-\mathbf{\Pi}\mathbf{\Lambda})^{-1}\mathbf{\Pi}\mathbf{\Lambda}'<\infty.
\end{equation}
From \eqref{beta}, it is clear that each element in $\sum_{i=2}^{\infty}\check{\beta}_i$ is positive, and thus the matrix $\mathbf{I}-\sum_{i=2}^{\infty}\check{\beta}_i$ does not have an all-zero row.
Therefore, $\vec{D}$ is bounded if the term $\sum_{i=2}^{\infty} \vec{\nu}(i)  \check{\beta}_i$ on the right-hand side of \eqref{D} is.

Using the inequalities~\eqref{inequality_alphabeta2} and \eqref{inequal_c}, after some simplifications, we have
\begin{equation} \label{final}
\begin{aligned}
\sum_{i=2}^{\infty} \vec{\nu}(i)  \check{\beta}_i
\preceq
\kappa \kappa_2 \sum_{i=2}^{\infty} \rho^{2i}(\mathbf{A}) \vec{1} \left[\rho^{i-2}(\mathbf{\Pi}\mathbf{\Lambda})\right]_{B \times B} +
\kappa\kappa_3 \sum_{i=3}^{\infty}  (i-2) \rho^{2i}(\mathbf{A}) \vec{1} \left[\rho^{i-3}(\mathbf{\Pi}\mathbf{\Lambda})\right]_{B \times B}.
\end{aligned}
\end{equation}
Therefore, it is clear that the right-hand side of \eqref{final} is bounded if $\rho(\mathbf{\Pi}\mathbf{\Lambda}) \rho^{2}(\mathbf{A}) <1$,
yielding the proof of the bounded expected first passage cost with condition~\eqref{stability_condition_markov}. Similarly, we can verify that the expected first passage time
is also bounded, completing the proof.

    \balance
    
	\ifCLASSOPTIONcaptionsoff
	\newpage
	\fi

	\bibliographystyle{IEEEtran}
%	\bibliography{IEEEabrv,cite}

\end{document}